\pdfoutput=1
\documentclass[journal,draftcls,onecolumn,12pt,twoside]{IEEEtranTCOM}

\usepackage[utf8]{inputenc}

\usepackage{amsmath,amssymb,amsfonts,amsthm}
\allowdisplaybreaks
\usepackage{graphicx}
\usepackage{xcolor}
\usepackage{pgfplots}
\usepackage{tikz}
\usetikzlibrary{decorations.markings}
\usepackage{balance}

\usepackage[english]{babel}

\newtheorem{theorem}{{Theorem}}
\newtheorem{corollary}{{Corollary}}
\newtheorem{definition}{{Definition}}
\newtheorem{remark}{{Remark}}

\title{A Pre-Transformation Method to Increase the Minimum Distance of Polar-Like Codes}
\author{Samet Gelincik, \IEEEmembership{Member, IEEE}, Philippe Mary, \IEEEmembership{Member, IEEE}, Anne Savard, \IEEEmembership{Member, IEEE} and Jean-Yves Baudais\thanks{S. Gelincik, P. Mary and J.-Y. Baudais are with Univ Rennes, INSA Rennes, CNRS, IETR-UMR 6164, F-35000 Rennes, France}\thanks{A. Savard is with IMT Nord Europe, Institut Mines T\'el\'ecom, Centre for Digital Systems, F-59653 Villeneuve d’Ascq, France.} \thanks{This work has been partially supported by IRCICA, CNRS USR 3380, Lille, France and the French National Agency for Research (ANR) under grant ANR-16-CE25-0001 ARBurst. 
Part of the content of this paper has been submitted to ISIT 2022 \cite{long_vs}.}}

\usepackage{hyperref}

\begin{document}
	
	\maketitle
	
	\begin{abstract}
		Reed Muller (RM) codes are known for their good minimum distance. One can use their structure to construct polar-like codes with good distance properties by choosing the information set as the rows of the polarization matrix with the highest Hamming weight, instead of the most reliable synthetic channels. However, the information length options of RM codes are quite limited due to their specific structure. In this work, we present sufficient conditions to increase the information length by at least one bit for some underlying RM codes and in order to obtain pre-transformed polar-like codes with the same minimum distance than lower rate codes.The proofs give a constructive method to choose the row triples to be merged together to increase the information length of the code and they follow from partitioning the row indices of the polar encoding matrix with respect to the recursive structure imposed by the binary representation of row indices. Moreover, our findings are combined with the method presented in \cite{ICC_row_merging} to further reduce the number of minimum weight codewords. Numerical results show that the designed codes perform close to the meta-converse bound at short blocklengths and better than the polarization-adjusted-convolutional polar codes with the same parameters.
	\end{abstract}
	
	\begin{IEEEkeywords}
		Polar codes, Reed Muller codes, minimum distance, finite block length.
	\end{IEEEkeywords}
	
	
	\section{Introduction}

New usages and services of 5G and beyond wireless systems, such as machine type communication or ultra-reliable low latency communications are pushing the limits of channel coding by requiring efficient error correcting codes at short to moderate block lengths. Indeed, These use-cases involve communicating objects that either occasionally transmit short packets at low power transmission to increase the device lifetime or because to meet stringent latency constraint~\cite{8594709}.
	
Polar codes, the first provably asymptotically capacity achieving error correcting codes over binary input memoryless channels \cite{polar} with explicit construction, are currently used over the control channels of 5G networks \cite{Standard5G}. They also are envisioned for ultra-reliable low-latency communications and massive machine-type communications \cite{LAND_5G} thanks to their low complexity successive cancellation based decoder. Unfortunately, standard polar codes do not show outstanding performance at short-to-moderate block lengths due to their poor minimum distance and a non-complete polarization \cite{9328621}. Several methods, such as enhanced-Bose–Chaudhuri–Hocquenghem subcodes \cite{Polar_sub} and low-weight-bit polar codes \cite{ISCAN}, have hence been proposed to improve their distance spectrum. Cyclic-redundancy-check (CRC) aided successive cancellation list (SCL) decoding, which boosts the performance by choosing the best decoding paths in a hierarchical tree, has been proposed in \cite{Tal_vardy} and the obtained performances were further enhanced by optimizing the CRC polynomial to improve the minimum distance of the obtained codes \cite{crc1,crc2}. The later was considered as the best code design in terms of Frame Error Rate (FER) up to the introduction of polarized adjusted convolutional (PAC) polar codes in \cite{arikan2019}.
	
PAC polar codes \cite{arikan2019}, by choosing the information set of the polar codes according to the Reed-Muller (RM) rule, i.e. the rows of the polarization matrix with the highest Hamming weights, perform very close to the second-order rate approximation of the binary-input additive white Gaussian noise in the short block length regime. It is a special case of the convolutional pre-transformation with an upper-triangular matrix, which has been proven not to reduce the minimum distance of underlying RM code while reducing the number of minimum weight codewords if properly designed \cite{2019pretransformed}. Since polar codes are tailored for a given channel, authors in \cite{8680016} proposed a genetic algorithm to obtain the frozen set that minimizes the bit or block error rate of the code over additive white Gaussian noise channel and Rayleigh channel. The obtained polar code does not require CRC and the proposed method either outperform CRC-aided polar codes occasionally or performs equally well in terms of bit and block error rate. Since the construction proposed in \cite{8680016} is signal to noise ratio dependent, authors in \cite{9389787} proposed a genetic algorithm, extended from \cite{8680016}, that relies on the distance spectrum of the code rather than an error rate criterion. This technique allows to achieve the normal approximation of the second-order rate in Gaussian noise channel. However, genetic algorithms are rather complex and strongly depend on the population size and its initial state. Contrarily, in our work, we give an explicit method to increase the minimum distance while adding additional information bits on the underlying RM code.
	
In \cite{ICC_row_merging}, we proposed a method to decrease the number of low weight codewords compared to RM and PAC codes in the short block length regime. The approach relies on encoding some information bits by the sum of two or three rows of the polar encoding matrix. The pairs and triplets of the merged rows are determined via the connection between the binary representation of the selected row indices and their common $1$ bit-positions. Notably, our designed codes achieve the same performance as PAC polar codes of the same parameters without extra computational complexity.
	
In this paper, we extend our previously proposed method \cite{ICC_row_merging} and state sufficient conditions to increase the information length of some polar-like codes, i.e. increasing the rate for a fixed given codeword length, where the information set is built according to the RM rule. Moreover, we explicitly give the corresponding pre-transformation matrix to sustain the same minimum distance as the RM code. The analysis is conducted by extending the method introduced in \cite{9184131} that partitions the row indices of the encoding matrix according to the indices of their binary representation. Numerical simulations show that our proposed code design outperforms the PAC polar codes with the same parameters in terms of FER, and performs close to the meta-converse (MC) bound thanks to the decrease in the number of minimum weight codewords achieved by our previously proposed algorithm in \cite[Algortihm 1]{ICC_row_merging}.

The rest of the paper is organized as follows. Section \ref{sec:prel} introduces the notations and the important definitions that are used in the proofs of our findings. Section \ref{sec:addinfbit} presents our main contribution with the statement of three theorems. Theorems \ref{thm:main_min_dis_inc_three_rows} and \ref{thm:main_min_dis_inc_with_intrsctn_three_rows} state the conditions for which the row merging increases the Hamming weight of the resulting row for the case where there is no common 1-bit position among the binary representations of the row indices and when there are such common positions, respectively. Theorem \ref{thm:Last_ensemble} is the main result of the paper and claims that is possible to increase the information length without decreasing the minimum distance of the code. In Section \ref{sec:codeconstruc}, we present our proposed polar-like code design which increases the number of information bits while exhibiting a high minimum distance. Section \ref{sec:simul} compares the performance of our proposed design with PAC polar codes and with the meta-converse bound. Finally, conclusions are drawn in Section \ref{sec:conclusion}.
	
	

	\section{Preliminaries} \label{sec:prel}
	
	\subsection{Notations}
	
	The positions of elements in a vector of length $N$ is indexed from $0$ to $N-1$. Any vector of length $N$ is considered as a row vector and is denoted by $\mathsf{x}$ or $\mathsf{x}^{N-1}$. The $j^{\text{th}}$ entry of the vector $\mathsf{x}$ is denoted as $x_j$. The set of positive integers is $\mathbb{N}$ and the binary field is $\mathbb{F}_2$. The set of integers from $j$ to $k-1$ is represented by $[j,k)$ or $[j,k-1]$. Uppercase calligraphic letters, such as $\mathcal{A}$, are reserved to index sets. Any index set is sorted in the ascending order and $\mathcal{A}(i)$, $i \in [0,|\mathcal{A}|)$ denotes the $i$-th element of $\mathcal{A}$. Specifically, we set $\mathcal{N}:=[0,N)$. For any given two index sets $\mathcal{A}$ and $\mathcal{B}$, $\mathcal{A} \succ \mathcal{B} $ denotes that any element of $\mathcal{A}$ is larger than any element of $\mathcal{B}$, i.e., $\mathcal{A}(i)>\mathcal{B}(j)$ $\forall i\in[0,|\mathcal{A}|)$ and $\forall j\in[0,|\mathcal{B}|)$.
	For a given binary vector $\mathsf{x} \in\mathbb{F}_2^{1\times N}$ and index set $\mathcal{A}\subset\mathcal{N}$, $\mathsf{x}_{\mathcal{A}}$ denotes the vector consisting of the elements of $\mathsf{x}$ at the positions indexed by $\mathcal{A}$. The matrices are denoted by uppercase sans serif font, e.g., $\mathsf{G}$. Uppercase boldface letters denote set of vectors, e.g., $\mathbf{C}$.
	The indicator function is $\mathbb{I}\{\cdot\}$. The sets $\mathcal{P}_1(\cdot)$ and $\mathcal{P}_0(\cdot)$ denote the indices of $1$'s and $0$'s of a given vector, respectively.
	
	
	For any $0\leq j<2^n$, its $n$-bit binary representation is denoted by the vector $\mathsf{b}_j^n$, or $\mathsf{b}_j$ if it is clear enough from the context. The $\ell$-th bit position of $\mathsf{b}_j$ is denoted by $ b_{j,\ell}$, $0\leq \ell < n $ and the indexing is started from the least significant bit, which is placed at the rightmost position. The number of $1$'s and $0$'s in a vector is represented by $i_1(\cdot)$ and $i_0(\cdot)$, respectively. 
	
	The operator $\bar{\cup}$ represents the element-wise 'OR' operation of binary vectors such that, for all $(j_1, j_2) \in \left[0, 2^n\right)^2$:
	\begin{align}
		b_{j_1,\ell} \bar{\cup} b_{j_2,\ell} =1, \ \text{if} \; \; & b_{j_1,\ell}=1 \ \text{or} \ b_{j_2,\ell}=1
	\end{align}
	The operator $\bar{\cap}$ represents the element-wise 'AND' operation of binary vectors such that
	\begin{align}
		b_{j_1,\ell} \bar{\cap} b_{j_2,\ell} =1, \ \text{if} \ b_{j_1,\ell}=b_{j_2,\ell}=1 
	\end{align}
	The operator $\oplus$ denotes binary addition in $\mathbb{F}_2$.
	
	\subsection{Properties of the Polar Encoding Matrix}
	
	For any given $N=2^n, \, n\in \mathbb{N}$, the polarization matrix is $\mathsf{G}=\mathsf{G}_2^{\otimes n}$ where
	\begin{align}
		\mathsf{G}_2:= \begin{bmatrix}
			1 & 0 \\
			1 & 1 
		\end{bmatrix}
	\end{align}
	is the corresponding kernel matrix and $\otimes$ is Kronecker pro\-duct. The $j$th row $\mathsf{g}_j$ of $\mathsf{G}$ can be represented by
	\begin{align}
		\mathsf{g}_j=\hat{\mathsf{g}}_{b_{j,n-1}}\otimes \hat{\mathsf{g}}_{b_{j,n-2}}\otimes \cdots \otimes \hat{\mathsf{g}}_{b_{j,0}} \label{Eq:g_rows0}
	\end{align}
	where $\hat{\mathsf{g}}_0=[1 \ \ 0]$ and $\hat{\mathsf{g}}_1=[1 \ \ 1]$.
	By \eqref{Eq:g_rows0}, for a given $n\in \mathbb{N}$, the $j$th row of $\mathsf{G}$ can be divided into $n$ disjoint regions, i.e.
	\begin{align}\label{Eq:reg_div}
		\mathsf{r}_{j,\ell} =
		\begin{cases}
			\mathsf{0}^{2^\ell-1} & \text{if } b_{j,\ell}=0\\
			[\mathsf{r}_{j,0} \mathsf{r}_{j,1} \cdots \mathsf{r}_{j,\ell-1}] & \text{if } b_{j,\ell}=1
		\end{cases}
	\end{align}
	for $\ell \in [1,n)$ and $\mathsf{r}_{j,0}=\hat{\mathsf{g}}_{b_{j,0}}$.
	Close inspection of the recursive nature of $\mathsf{r}_{j,\ell}$ reveals that each bit position $\ell\in [0,n)$ of $\mathsf{b}_j$ is associated with a set of positions at $\mathsf{g}_j$ denoted by the index set $\mathcal{M}_{\ell}\subset \mathcal{N}$
	\begin{align}\label{}
		\mathcal{M}_{\ell}\hspace{-0.08cm}:=\{k: b_{k,\ell}=1, k\in \mathcal{N} \}
	\end{align}
	and $\mathcal{M}^{c}_{\ell}:=\mathcal{N}\setminus \mathcal{M}_{\ell}$.
	The fact that $\mathsf{g}_{j,\mathcal{M}_{\ell}}=\mathsf{0}^{N/2-1}$ if $b_{j,\ell}=0$ imposes that $\mathsf{g}_{j,\mathcal{M}_{\ell}^{c}}$ is independent from the value of $b_{j,\ell}$ \cite{ICC_row_merging}. The following definition highlights this fact.
	
	\begin{definition} \label{def:1}
		The projection
		of a row $\mathsf{g}_j$ of the polar encoding matrix onto indices of $\mathcal{M}_{\ell}^{c}$ is denoted by $\mathsf{g}_{j}^{\ell}$ and given as
		\begin{align}\label{eq:proj_v0}
			\mathsf{g}_{j}^{\ell}&:=\hat{\mathsf{g}}_{b_{j,n-1}}\otimes\cdots \otimes \hat{\mathsf{g}}_{b_{j,\ell+1}}\otimes\hat{\mathsf{g}}_{b_{j,\ell-1}} \cdots \otimes \hat{\mathsf{g}}_{b_{j,0}}
		\end{align}
	\end{definition}
	
	
	Note that, by \eqref{Eq:reg_div}, $\mathsf{b}_{j,\ell} = 1$ imposes that $[\mathsf{r}_{j,0} \mathsf{r}_{j,1} \cdots \mathsf{r}_{j,\ell-1}]$ is copied to $\mathsf{r}_{j,\ell}$ and $\mathsf{r}_{j,t>\ell}$ is obtained with respect to corresponding bit values. Hence, the projection of $\mathsf{g}_j$ onto $\mathcal{M}_{\ell}$ is the same as $\mathsf{g}_{j,\mathcal{M}_{\ell}^{c}}$ if $b_{j,\ell}=1$
	\begin{align}\label{eq:proj_v1}
		\mathsf{g}_{j,\mathcal{M}_{\ell}}=\begin{cases}
			\mathsf{0}^{\frac{N}{2}-1} & \text{if } b_{j,\ell}=0\\
			\mathsf{g}_{j}^{\ell} & \text{if } b_{j,\ell}=1
		\end{cases}
	\end{align}
	
	The following definition is the generalization of Definition~\ref{def:1}. 
	
	\begin{definition}\label{def:2}
		The projection of row $\mathsf{g}_j$ of the polar encoding matrix onto $\displaystyle \cap_{\ell \in \mathcal{B}}\mathcal{M}^{c}_{\ell}$ is denoted by $\mathsf{g}_j^{\mathcal{B}}$ and $\mathsf{g}_j|\cap_{\ell \in \mathcal{B}}\mathcal{M}^{c}_{\ell}$, and is given as
		\begin{align}\label{eq:subset_proj0}
			\mathsf{g}^{\mathcal{B}}_{j}:&=\mathsf{g}_{j}|\displaystyle{\cap_{\ell \in \mathcal{B}}\mathcal{M}^{c}_{\ell}}\nonumber\\
			&=\hat{\mathsf{g}}_{b_{j,\mathcal{W}(|\mathcal{W}|-1)}}\otimes \hat{\mathsf{g}}_{b_{j,\mathcal{W}(|\mathcal{W}|-2)}}\otimes \cdots \otimes \hat{\mathsf{g}}_{b_{j,\mathcal{W}(0)}}
		\end{align}
		where $\mathcal{W}:=[0,n)\setminus\mathcal{B}$.
	\end{definition}
	
	Note that, similar to \eqref{eq:proj_v1}, for any subset $\mathcal{B}_0\subset \mathcal{B}$, the projection of $\mathsf{g}_j$ onto $\cap_{\ell\in \mathcal{B}_0}\mathcal{M}_{\ell}\cap_{\hat{\ell}\in \mathcal{B}/\mathcal{B}_0}\mathcal{M}^{c}_{\hat{\ell}}$ is given by
	\begin{align}\label{eq:subset_proj1}
		\mathsf{g}_{j}|\cap_{\ell\in \mathcal{B}_0}\mathcal{M}_{\ell}\cap_{\hat{\ell}\in \mathcal{B}\backslash \mathcal{B}_0}\mathcal{M}^{c}_{\hat{\ell}}=
		\left\{\begin{array}{l@{\ }l}
			\mathsf{0}^{\frac{N}{|\mathcal{B}|}-1} & \text{if } \bar{\cap}_{\ell\in \mathcal{B}_0}b_{j,\ell}=0\\
			\mathsf{g}^{\mathcal{B}}_{j} & \text{if } \bar{\cap}_{\ell\in \mathcal{B}_0}b_{j,\ell}=1
		\end{array}\right.
	\end{align}
	
	\subsection{Row Merging Pre-transformed Polar-like Codes and RM Codes}
	
	A polar-like code $(N=2^n,k)\in \mathbb{N}^2$, is constructed as
	\begin{align}
		\mathbf{C}=\{\mathsf{c}=\mathsf{u}\mathsf{G}:\mathsf{u}\in\mathbb{F}^{n}_2, \mathsf{u}_{\mathcal{F}}=\mathsf{0} \}
	\end{align}
	where $\mathcal{F}$ is the index set of the frozen bit positions, and $\mathcal{A}=\mathcal{N}\setminus \mathcal{F}$ is the information set. For classical polar codes under SC decoding, the set $\mathcal{A}$ is the set of the most reliable bit sub-channels \cite{polar}. However in this paper, we allow to choose the information set differently. From this perspective, a RM$(n,r)$ code of degree $r$ can be seen as a polar-like code of information set 
	\begin{align}
		\mathcal{A}= \bigcup_{p=n-r}^{n}\mathcal{N}_{p},\quad \mathcal{N}_{p}:=\{t:i_1(\mathsf{b}_t)=p, t\in\mathcal{N}\}.
	\end{align}
	In \cite{urbanke}, the minimum distance of a polar-like code is given by
	\begin{align}\label{eq:low_bound_d}
		d(\mathbf{C})= \min_{i\in \mathcal{A}}i_1(\mathsf{g}_i)\overset{(a)}{=}2^{\min_{i\in \mathcal{A}}i_1(\mathsf{b}_i)}
	\end{align}
	where (a) is due to \cite[Theorem 2]{ICC_row_merging}.
	
	The pre-transformed polar-like codes \cite{2019pretransformed} is obtained through a pre-transformation matrix $\mathsf{T}\in \mathbb{F}^{N \times N}_2$ %
	\begin{align}
		\mathbf{C}_{\mathbf{P}}=\{\mathsf{c}=\mathsf{u}\mathsf{T}\mathsf{G}:\mathsf{u}\in\mathbb{F}^{n}_2,\mathsf{u}_{\mathcal{F}}=\mathsf{0} \}
	\end{align}
	where $\mathsf{T}$ is an upper triangular matrix with $\mathsf{T}_{i,i}=1$, $i\in \mathcal{N}$ and $\mathcal{F}_d:=\{j: T_{i,j}=1, i\in \mathcal{N}, j>i\}$ is the set of dynamic frozen bits. If $\mathsf{T}$ is restricted such that $ |\{i:\mathsf{T}_{i,j}, i\in \mathcal{N}\}|\in\{1,2\}$
	$\forall j\in \mathcal{F}_d$, then $\mathsf{T}$ turns out to be a row merging pre-transformation matrix since some information bits are encoded with more than one row of the polarization matrix but any frozen row can be associated with at most one information row 
	\begin{equation}
		\mathsf{c}=\mathsf{u}\mathsf{T}\mathsf{G}=\mathsf{u}\tilde{\mathsf{G}}
	\end{equation}
	with
	\begin{align}
		\tilde{\mathsf{g}}_i=\mathsf{g}_i\bigoplus_{j\in \mathcal{P}_1(\tilde{\mathsf{t}}_i)\backslash i}\mathsf{g}_j
	\end{align}
	where $\tilde{\mathsf{t}}_i$ is the $i$-th row of $\mathsf{T}$.

	\section{Adding Information Bits to RM Information Set by Sustaining the Same Minimum Distance}\label{sec:addinfbit}
	
	In this section, we present how to obtain triples of polarization matrix rows to keep the same minimum distance as the underlying RM code and state the size of information length increment for some given parameters. 
	Let $\mathcal{T}\subseteq \mathcal{N}$ be any subset of row indices of the polarization matrix $\mathsf{G}$ and $i \in \mathcal{N}\setminus \mathcal{T}$. Then, by $\mathsf{g}_{\mathcal{T}}$ and $\mathsf{g}_{\{i,\mathcal{T}\}}$, we denote
	\begin{align}
		\mathsf{g}_{\mathcal{T}}=\bigoplus_{t\in \mathcal{T}}\mathsf{g}_t \hspace{0.5cm}\text{   and   } \hspace{0.5cm} \mathsf{g}_{\{i,\mathcal{T}\}}= \mathsf{g}_i\oplus \mathsf{g}_{\mathcal{T}}
	\end{align}
	
	\subsection{Preliminary Theorems}
	
	
	For the sake of completeness, we first state Theorem 2 of \cite{ICC_row_merging} and give a corollary that will be exploited later on in this paper.

\begin{theorem}{\cite[Theorem 2]{ICC_row_merging}}\label{th:Hamm_of_sum0000}
Let $\mathcal{T}\subseteq \mathcal{N}$ be any subset of row indices of polar-like code generator matrix $\mathsf{G}_N$. Then, the Hamming weight of the sum of the rows $\mathsf{g}_j, \ j\in \mathcal{T}$ is given by
\begin{align}\label{Eq:Hamm_of_sum0000}
    i_1\big(\mathsf{g}_{\mathcal{T}})=\sum_{w=1}^{|\mathcal{T}|}(-2)^{w-1}\sum_{\mathcal{T}^w\subset \mathcal{T}}2^{i_1\big(\bar{\bigcap}_{j\in \mathcal{T}^{w}}\mathsf{b}_j\big)}
\end{align}
where $\mathcal{T}^{w}$ is any subset of $\mathcal{T}$ with $w$ elements.
\end{theorem}
	
	
	\begin{corollary}\label{cor:thm_1}
		Let $\Pi:\mathbb{F}_2^{n}\mapsto \mathbb{F}_2^{n}$ be a permutation on binary representations of $j\in \mathcal{N}$ and $\tilde{\mathcal{T}}$ be the index set obtained by applying permutation $\Pi$ to the binary representations of elements of $\mathcal{T}$: $\mathsf{b}_{\tilde{j}}=\Pi(\mathsf{b}_{j})$, $j\in \mathcal{T}$ and $\tilde{j}\in \tilde{\mathcal{T}}$. 
		Then, 
		\begin{align}
			i_1(\mathsf{g}_{\mathcal{T}})=i_1(\mathsf{g}_{\tilde{\mathcal{T}}})
		\end{align}
	\end{corollary}
	
	\begin{proof} The number of common $1-$bits will not change with permutation for any subset $\mathcal{T}^{w}\subset \mathcal{T}$, i.e.,
		\begin{align}
			i_1(\bar{\cap}_{j\in\mathcal{T}^{w} }\mathsf{b}_{j})&=i_1(\bar{\cap}_{j\in\mathcal{T}^{w} }\Pi(\mathsf{b}_{j}))=i_1(\bar{\cap}_{\tilde{j}\in\Pi_{\mathcal{T}^{w}} }\mathsf{b}_{\tilde{j}})=i_1(\bar{\cap}_{\tilde{j}\in\tilde{\mathcal{T}}^{w} }\mathsf{b}_{\tilde{j}})
		\end{align}
		then, by \eqref{Eq:Hamm_of_sum0000}, the Hamming weight does not change.
	\end{proof}

	The following theorem is also used to obtain subsequent results of this paper. It basically states that for any given set of rows of the polarization matrix, the Hamming weight of the sum of all rows is lower bounded by the maximum Hamming weight of the sum of a subset of rows whose binary representations are zero at the corresponding binary indices.
	
	
	\begin{theorem}\label{thm:main_row_merg0}
		For any given $\mathcal{T}\subseteq \mathcal{N}$ the Hamming weight of $\mathsf{g}_{\mathcal{T}}$ is lower bounded by
		\begin{align}
			i_1(\mathsf{g}_{\mathcal{T}})\geq \max_{\ell\in [0,n)}i_1(\mathsf{g}_{\mathcal{T}^{0}_{\ell}})
		\end{align}
		where $\mathcal{T}^{0}_{\ell}:=\{k:b_{k,\ell}=0, k\in \mathcal{T}\}$.
	\end{theorem}
	
	\begin{proof} 
		For any $\mathsf{u},\mathsf{v}\in \mathbb{F}^{1\times N}_2$, we have
		\begin{align}\label{eq:u_v}
			i_1(\mathsf{u} \oplus \mathsf{v})+i_1(\mathsf{v})&= i_1(\mathsf{u})+i_1(\mathsf{v})-2\cdot i_1(\mathsf{u} \bar{\cap} \mathsf{v})+i_1(\mathsf{v})\nonumber \\
			&= i_1(\mathsf{u})+2\cdot ( \underbrace{i_1(\mathsf{v})-i_1(\mathsf{u} \bar{\cap} \mathsf{v})}_{\geq 0 }) \geq i_1(\mathsf{u}).
		\end{align}
		Then, note that for any $j \in \mathcal{N} $
		\begin{align}\label{eq:in_proof3}
			i_1(\mathsf{g}_j)=\begin{cases}
				i_1(\mathsf{g}^{\ell}_j) & \text{if } b_{j,\ell}=0\\
				2\cdot i_1(\mathsf{g}^{\ell}_j) & \text{if } b_{j,\ell}=1
			\end{cases}
		\end{align}
		for any $\ell \in [0,n)$ due to \eqref{eq:proj_v0} and \eqref{eq:proj_v1}. Therefore, for any $\ell\in [0,n)$ we can write 
		\begin{align}
			i_1(\mathsf{g}_{\mathcal{T}}) &=(\mathsf{g}_{\mathcal{T}}|\mathcal{M}^{c}_{\ell})+(\mathsf{g}_{\mathcal{T}}|\mathcal{M}_{\ell}) \overset{(a)}{=} i_1(\bigoplus_{j\in \mathcal{T}}\mathsf{g}^{\ell}_{j})+i_1(\bigoplus_{j\in \mathcal{T}}\mathsf{g}^{\ell}_{j}\mathbb{I}\{b_{j,\ell}=1\})\nonumber \\
			&=i_1(\bigoplus_{j\in \mathcal{T}}\mathsf{g}^{\ell}_{j}\mathbb{I}\{b_{j,\ell}=0\}\bigoplus_{j\in \mathcal{T}}\mathsf{g}^{\ell}_{j}\mathbb{I}\{b_{j,\ell}=1\})+i_1(\bigoplus_{j\in \mathcal{T}}\mathsf{g}^{\ell}_{j}\mathbb{I}\{b_{j,\ell}=1\})\nonumber \\
			&\overset{(b)}{\geq} i_1(\bigoplus_{j\in \mathcal{T}}\mathsf{g}^{\ell}_{j}\mathbb{I}\{b_{j,\ell}=0\})
			= i_1(\bigoplus_{j\in \mathcal{T}^{0}_{\ell}}\mathsf{g}^{\ell}_{j}) 
			\overset{(c)}{=} i_1(\bigoplus_{j\in \mathcal{T}^{0}_{\ell}}\mathsf{g}_{j})
		\end{align}
		where (a) is due to \eqref{eq:proj_v0} and \eqref{eq:proj_v1}, (b) is due to \eqref{eq:u_v} and (c) is due to \eqref{eq:in_proof3}.
	\end{proof}
	


	\begin{theorem}\label{thm:main_min_dis_inc_three_rows}
		Let $\mathbf{C}$ be a polar-like code with information set $\mathcal{A}=\bigcup_{p=\ell+1}^{n}\mathcal{N}_p$. and $(i,j,k)$ be a triple such that $(i,j)\in \mathcal{N}_{\ell}, \ell \geq 2, \; k\in \mathcal{N}_2\text{ and } i_1(\mathsf{b}_i\bar{\cap}\mathsf{b}_j)=i_1(\mathsf{b}_i\bar{\cap}\mathsf{b}_k)=i_1(\mathsf{b}_j\bar{\cap}\mathsf{b}_k)=0$. Moreover, let $\mathbf{\bar C}$ be another polar-like code that encodes an additional information bit by $\mathsf{g}_i\oplus\mathsf{g}_j\oplus\mathsf{g}_k $, i.e.
		\begin{align}
			\bar{\mathbf{C}}:=\{\mathbf{C}\}\cup \{\mathsf{c}: \mathsf{c}=\mathsf{g}_{\{i,j,k\}}\oplus \mathsf{g}_{\mathcal{T}}, \mathcal{T}\subseteq\mathcal{A}\}.
		\end{align}
		Then, the minimum distance of $\mathbf{\bar C}$ is the same as $\mathbf{C}$, i.e.
		\begin{align}
			d(\bar{\mathbf{C}})&=\min\{d(\mathbf{C}), \min_{\mathcal{T}\subseteq \mathcal{A}}	i_1(\mathsf{g}_{\{i,j,k\}}\oplus \mathsf{g}_{\mathcal{T}})\}\nonumber \\&=d(\mathbf{C})=2^{\ell+1} 
		\end{align}
	\end{theorem}
	
	\begin{proof}
		The proof is given in Appendix~\ref{app:thm3}.
	\end{proof}
	
	
	\subsection{Merging Three Rows with Common 1-bit Positions}
	
	The following theorem is a generalization of Theorem~\ref{thm:main_min_dis_inc_three_rows} and states the sufficient conditions on the rows of a triple with some common $1-$bit positions in their binary representations, to be merged together such that the minimum distance of the underlying RM code is preserved.
	
	\begin{theorem}\label{thm:main_min_dis_inc_with_intrsctn_three_rows}
		Let $\mathbf{C}$ be a polar-like code with information set $\mathcal{A}=\bigcup_{p=\ell+1}^{n}\mathcal{N}_p$. and $(i,j,k)$ be a triple such that $\mathcal{P}_1(\mathsf{b}_i\bar{\cap}\mathsf{b}_j)=\mathcal{P}_1(\mathsf{b}_i\bar{\cap}\mathsf{b}_k)=\mathcal{P}_1(\mathsf{b}_j\bar{\cap}\mathsf{b}_k)\neq \emptyset$, $(i,j)\in \mathcal{N}_{\ell}$, $k\in \mathcal{N}_{i_1(\mathsf{b}_i\bar{\cap}\mathsf{b}_j)+2}$, $\ell \geq i_1(\mathsf{b}_k)$. Let the code $\bar{\mathbf{C}}$ be:
		\begin{align}
			\bar{\mathbf{C}}:=\{\mathbf{C}\}\cup \{\mathsf{c}: \mathsf{c}=\mathsf{g}_{\{i,j,k\}}\oplus \mathsf{g}_{\mathcal{T}},\, \mathcal{T} \subseteq \mathcal{A}\}
		\end{align}
		Then,
		\begin{align}
			d(\bar{\mathbf{C}})=d(\mathbf{C})=2^{\ell+1} 
		\end{align}
	\end{theorem}
	
	\begin{proof}
		Since 
		\begin{align}
			d(\bar{\mathbf{C}})&=\min\{d(\mathbf{C}), \min_{\mathcal{T}\subseteq \mathcal{A}}	i_1(\mathsf{g}_{\{i,j,k\}}\oplus \mathsf{g}_{\mathcal{T}})\} %
		\end{align}
		it is sufficient to prove the following statement 
		\begin{align}\label{eq:eqt_thm40}
			i_1(\mathsf{g}_{\{i,j,k\}}\oplus \mathsf{g}_{\mathcal{T}})&\geq 2^{\ell+1}, \; \; \forall \mathcal{T}\subseteq \mathcal{A}.
		\end{align}
		
		For any $\mathcal{T}\subseteq \mathcal{A}$, the index set can be divided into two subsets such that
		\begin{align}\label{eq:in_thm_subset01}
			\tilde{\mathcal{T}}:=\{t:\mathcal{P}_1(\mathsf{b}_t)\cap\mathcal{P}_0(\mathsf{b}_i\bar{\cup}\mathsf{b}_j\bar{\cup}\mathsf{b}_k)\neq \emptyset, t\in \mathcal{T} \}
		\end{align}
		and $\hat{\mathcal{T}}=\mathcal{T}\setminus \tilde{\mathcal{T}}$.
		Then, 
		\begin{align}
			i_1(\mathsf{g}_{\{i,j,k\}}\oplus\mathsf{g}_{\hat{\mathcal{T}}}\oplus\mathsf{g}_{ \tilde{\mathcal{T}}})  &\overset{(a)}{\geq} \hspace{-0.10cm}\max_{p_0\in \mathcal{P}_0(\mathsf{b}_i\bar{\cup}\mathsf{b}_j\bar{\cup}\mathsf{b}_k)}\hspace{-0.10cm}i_1(\mathsf{g}_{\{i,j,k\}}\oplus \mathsf{g}_{\hat{\mathcal{T}}} \bigoplus_{t\in \tilde{\mathcal{T}}}\mathsf{g}_t \mathbb{I}\{b_{t,p_0}=0\})\nonumber\\
			&\overset{(b)}{\geq} \hspace{-0.20cm}\max_{p_1\in \mathcal{P}_0(\mathsf{b}_i\bar{\cup}\mathsf{b}_j\bar{\cup}\mathsf{b}_k)\setminus p_0}\hspace{-0.45cm} i_1(\mathsf{g}_{\{i,j,k\}}\oplus \mathsf{g}_{\hat{\mathcal{T}}}\bigoplus_{t\in \tilde{\mathcal{T}}}\mathsf{g}_t \mathbb{I}\{b_{t,p_0}={b_{t,p_1}=0}\})\nonumber \\
			&\hspace{0.15cm}\vdots \nonumber \\
			&\overset{(c)}{\geq} i_1(\mathsf{g}_{\{i,j,k\}}\oplus \mathsf{g}_{\hat{\mathcal{T}}} \bigoplus_{t\in \tilde{\mathcal{T}}}\mathsf{g}_t \mathbb{I}\{b_{t,p_0} \hspace{-0.10cm}= \hspace{-0.05cm}{b_{t,p_1} \hspace{-0.10cm}= \hspace{-0.10cm}\cdots \hspace{-0.10cm} = \hspace{-0.05cm} b_{t,p_{n-2 \ell -1}} \hspace{-0.10cm}= \hspace{-0.05cm}0}\})\nonumber \\
			&\overset{(d)}{=}i_1(\mathsf{g}_{\{i,j,k\}}\oplus \mathsf{g}_{\hat{\mathcal{T}}})
		\end{align}
		where $\{p_0,p_1,\cdots,p_{n-2\ell-1}\}=\mathcal{P}_0(\mathsf{b}_i\bar{\cup}\mathsf{b}_j\bar{\cup}\mathsf{b}_k)$, where (a), (b) and (c) follow from the repeated application of Theorem~\ref{thm:main_row_merg0}, and (d) comes from \eqref{eq:in_thm_subset01}, which implies that there is no $t\in \tilde{\mathcal{T}}$ such that $\mathcal{P}_1(\mathsf{b}_t)\cap \{p_0,p_1,\cdots,p_{n-2\ell-1}\}=\emptyset$.
		This means that the Hamming weight of $\mathsf{g}_{\{i,j,k,\mathcal{T}\}}$ is lower bounded by the Hamming weight of $\mathsf{g}_{\{i,j,k,\hat{\mathcal{T}}\}}$. Therefore, in the following, we will proceed the proof for $\hat{\mathcal{T}}$. 
		%
		
		Now, assume that $\mathcal{W}=\mathcal{P}_1(\mathsf{b}_i\bar{\cap}\mathsf{b}_j)$, i.e., $\mathcal{W}=\mathcal{P}_1(\mathsf{b}_i\bar{\cap}\mathsf{b}_j\bar{\cap}\mathsf{b}_k)$ as well, by assumption. Then, by partitioning the row indices of the polar encoding matrix respect to binary bit positions $\mathcal{W}$, we obtain the following expression
		\begin{align}
			i_1(\mathsf{g}_{\{i,j,k\}}\hspace{-0.12cm} \oplus \mathsf{g}_{\hat{\mathcal{T}}}) &=i_1(\mathsf{g}_{\{i,j,k\}} \hspace{-0.12cm}\oplus\hspace{-0.07cm} \mathsf{g}_{\hat{\mathcal{T}}}| \mathcal{M}^{c}_{\mathcal{W}(|\mathcal{W}|-1)}\hspace{-0.08cm}\cap \hspace{-0.08cm} \mathcal{M}^{c}_{\mathcal{W}(|\mathcal{W}|-2)}\hspace{-0.08cm}\cap \hspace{-0.08cm}\cdots \hspace{-0.08cm}\cap \hspace{-0.08cm}\mathcal{M}^{c}_{ \mathcal{W}(0)}) \nonumber \\ & +i_1(\mathsf{g}_{\{i,j,k\}} \hspace{-0.12cm} \oplus \hspace{-0.08cm} \mathsf{g}_{\hat{\mathcal{T}}}| \mathcal{M}^{c}_{\mathcal{W}(|\mathcal{W}|-1)}\hspace{-0.08cm}\cap\hspace{-0.08cm} \mathcal{M}^{c}_{\mathcal{W}(|\mathcal{W}|-2)}\hspace{-0.08cm}\cap\hspace{-0.08cm} \cdots \hspace{-0.08cm}\cap\hspace{-0.08cm} \mathcal{M}_{ \mathcal{W}(0)})\nonumber \\ & \hspace{3.5cm}\vdots \nonumber \\ &+i_1(\mathsf{g}_{\{i,j,k\}}\hspace{-0.12cm}\oplus\hspace{-0.08cm} \mathsf{g}_{\hat{\mathcal{T}}}|\mathcal{M}_{\mathcal{W}(|\mathcal{W}|-1)}\hspace{-0.08cm}\cap \hspace{-0.08cm} \mathcal{M}_{\mathcal{W}(|\mathcal{W}|-2)}\hspace{-0.08cm}\cap\hspace{-0.08cm} \cdots\hspace{-0.08cm} \cap \hspace{-0.08cm}\mathcal{M}_{ \mathcal{W}(0)})\nonumber \\
			& \overset{(a)}{=} i_1(\mathsf{g}^{\mathcal{W}}_{\{i,j,k\}}\bigoplus_{t\in \hat{\mathcal{T}}} \mathsf{g}^{\mathcal{W}}_{t})\hspace{-0.05cm}+\hspace{-0.05cm} i_1(\mathsf{g}^{\mathcal{W}}_{\{i,j,k\}}\bigoplus_{t\in \hat{\mathcal{T}}} \mathsf{g}^{\mathcal{W}}_{t}\mathbb{I}\{b_{t,\mathcal{W}(0)}=1\}) \nonumber \\ & \hspace{3.5cm}\vdots \nonumber \\&+ i_1(\mathsf{g}^{\mathcal{W}}_{\{i,j,k\}} \bigoplus_{t\in \hat{\mathcal{T}}} \mathsf{g}^{\mathcal{W}}_{t}\mathbb{I}\{b_{t,\mathcal{W}(|\mathcal{W}|-1)}\hspace{-0.1cm}=\hspace{-0.1cm}b_{t,\mathcal{W}(|\mathcal{W}|-2)}\hspace{-0.1cm}=\hspace{-0.1cm}\cdots =b_{t,\mathcal{W}(0)}=1\})
		\end{align}
		where (a) is due to \eqref{eq:subset_proj0} and \eqref{eq:subset_proj1}. Since $\mathcal{P}_1(\mathsf{b}_i)\setminus \mathcal{P}_1(\mathsf{b}_i\bar{\cap}\mathsf{b}_j)$, $\mathcal{P}_1(\mathsf{b}_j)\setminus \mathcal{P}_1(\mathsf{b}_i\bar{\cap}\mathsf{b}_j)$ and $\mathcal{P}_1(\mathsf{b}_k)\setminus \mathcal{P}_1(\mathsf{b}_i\bar{\cap}\mathsf{b}_k)$ comply with the conditions of Theorem~\ref{thm:main_min_dis_inc_three_rows}, each term of the partition is lower bounded by $2^{\ell-|\mathcal{W}|+1}$. Then,
		\begin{align}
			i_1(\mathsf{g}_{\{i,j,k\}}\oplus \mathsf{g}_{\hat{\mathcal{T}}})&\geq2^{i_1(\mathsf{b}_i\bar{\cap} \mathsf{b}_j)}\cdot (2^{\ell-i_1(\mathsf{b}_i\bar{\cap} \mathsf{b}_j)+1})=2^{\ell+1}
		\end{align}
		where $|\mathcal{W}|=i_1(\mathsf{b}_i\bar{\cap} \mathsf{b}_j)$.
	\end{proof}
	
	In the following, we state the sufficient conditions to increase the information length by multiple bits for a fix codeword length. Thanks to the symmetry imposed by Corollary~\ref{cor:thm_1}, we apply a permutation ${\Pi}$ to any given row triple satisfying the conditions of Theorem~\ref{thm:main_min_dis_inc_with_intrsctn_three_rows} to have the following form
	\begin{align}\label{eq:set_greater1}
		\mathcal{P}_1(\mathsf{b}_i\bar{\cup} \mathsf{b}_j\bar{\cup} \mathsf{b}_k)	\succ	\mathcal{P}_0(\mathsf{b}_i\bar{\cup} \mathsf{b}_j\bar{\cup} \mathsf{b}_k)
	\end{align}
	and
	\begin{align}\label{eq:set_greater2}
		\mathcal{P}_1(\mathsf{b}_i\bar{\cap} \mathsf{b}_j\bar{\cap} \mathsf{b}_k) \succ 	\mathcal{P}_1(\mathsf{b}_i\bar{\cup} \mathsf{b}_j\bar{\cup} \mathsf{b}_k)\setminus	\mathcal{P}_1(\mathsf{b}_i\bar{\cap} \mathsf{b}_j\bar{\cap} \mathsf{b}_k)
	\end{align}	
	and 
	\begin{align}\label{eq:set_greater3}
		&\mathcal{P}_1(\mathsf{b}_k)\setminus \mathcal{P}_1(\mathsf{b}_i\bar{\cap}\mathsf{b}_j)\succ\mathcal{P}_1(\mathsf{b}_j)\setminus \mathcal{P}_1(\mathsf{b}_i\bar{\cap}\mathsf{b}_j),\nonumber \\
		&\mathcal{P}_1(\mathsf{b}_k)\setminus \mathcal{P}_1(\mathsf{b}_i\bar{\cap}\mathsf{b}_j)\succ\mathcal{P}_1(\mathsf{b}_i)\setminus \mathcal{P}_1(\mathsf{b}_i\bar{\cap}\mathsf{b}_j) 
	\end{align}
	and 
	\begin{align}\label{eq:set_greater4}
		&\mathcal{P}_1(\mathsf{b}_i)\setminus \mathcal{P}_1(\mathsf{b}_i\bar{\cap}\mathsf{b}_j) \nsucc\mathcal{P}_1(\mathsf{b}_j)\setminus \mathcal{P}_1(\mathsf{b}_i\bar{\cap}\mathsf{b}_j),\nonumber \\ 
		&\mathcal{P}_1(\mathsf{b}_j)\setminus \mathcal{P}_1(\mathsf{b}_i\bar{\cap}\mathsf{b}_j) \nsucc\mathcal{P}_1(\mathsf{b}_i)\setminus \mathcal{P}_1(\mathsf{b}_i\bar{\cap}\mathsf{b}_j)
	\end{align}
		
		Moreover, let
		$\Pi^{\theta}_{p}$ be a left-circular shift permutation on the index set of binary representation of $p\in \mathcal{N}$, with $\theta\in [0,\kappa], \kappa=t_0\cdot \mathbb{I}\{t_1>0\}+t_1$, $t_0=i_0(\mathsf{b}_i\bar{\cup} \mathsf{b}_j\bar{\cup} \mathsf{b}_k)$ and $t_1=i_1(\mathsf{b}_i\bar{\cap} \mathsf{b}_j\bar{\cap} \mathsf{b}_k)$. We have
		\begin{align}\label{eq:perm_38}
			b_{\Pi^{\theta}_{p},v}=b_{p,v-\theta+n \pmod n}
		\end{align} 
		The following theorem is the main result of this paper.
		\begin{theorem}\label{thm:Last_ensemble}
			Let $\mathbf{C}$ be a polar-like code with information set $\mathcal{A}=\bigcup_{p=\ell+1}^{n}\mathcal{N}_p$. Let $(i,j,k)$ be a triple satisfying the conditions of Theorem~\ref{thm:main_min_dis_inc_with_intrsctn_three_rows} and \eqref{eq:set_greater1}, \eqref{eq:set_greater2}, \eqref{eq:set_greater3}, \eqref{eq:set_greater4}.
			Let $\bar{\mathbf{C}}$ be a code obtained by encoding each of the extra $m \leq t_0+t_1+1$ information bits with a merged row triple $\mathsf{g}_{\{\Pi^{\theta}_{i},\Pi^{\theta}_{j},\Pi^{\theta}_{k}\}}$. Then,
			\begin{align}
				d(\bar{\mathbf{C}})=d(\mathbf{C})
			\end{align}
		\end{theorem}
		
		\begin{proof}
			The proof is given in Appendix~\ref{app:thm4}.
		\end{proof}
		The following section explains how Theorem \ref{thm:Last_ensemble} is used in order to increase the information length of a polar-like code with RM information set by preserving the minimum distance.

\begin{table}
	\centering
	\begin{tikzpicture}[x=1.6cm,y=.75cm]
		\draw (0,0) grid [step=1] (5,3);
		\draw (0,3) -- (1,2);
		\node at (0.5,2.5) [below left,inner sep=2pt] {\small$n$};
		\node at (0.5,2.35) [above right,inner sep=4pt] {\small$r$};
		\node at (0.5,1.5) {6};
		\node at (0.5,0.5) {7};
		\node at (1.5,2.5) {2};
		\node at (2.5,2.5) {3};
		\node at (3.5,2.5) {4};
		\node at (4.5,2.5) {5};
		\node at (1.5,1.5) {\footnotesize{$(1,23,16)$}};
		\node at (2.5,1.5) {\footnotesize{$-$}};
		\node at (3.5,1.5) {\footnotesize{$-$}};
		\node at (4.5,1.5) {\footnotesize{$-$}};
		\node at (1.5,0.5) {\footnotesize{$-$}};
		\node at (2.5,0.5) {\footnotesize{$(2,66,16)$}};
		\node at (3.5,0.5) {\footnotesize{$(1,100,8)$}};
		\node at (4.5,0.5) {\footnotesize{$-$}};
	\end{tikzpicture}
	\caption{Number of additional information bits $m$ that can be added on top of the information length $k$ for the minimum distance $d$ according to the recursion number $n$ and order $r$ of the underlying RM-polar code.}
	\label{Fig:Table_main}
\end{table}

\section{Code Construction} \label{sec:codeconstruc}


Let us consider a triple $(i,j,k)$ that satisfies the conditions of Theorem~\ref{thm:main_min_dis_inc_with_intrsctn_three_rows}, \eqref{eq:set_greater1}, \eqref{eq:set_greater2}, \eqref{eq:set_greater3} and \eqref{eq:set_greater4}. For any $m\in [1,t_0+t_1+1]$,
\begin{itemize}
    \item Each of $m-1$ triples,  i.e., $\left\{(i_0,j_0,k_0) , \cdots, (i_{m-2},j_{m-2},k_{m-2})\right\}$, corresponds to one of consecutive left-circular shifts of $(i,j,k)$.
	\item For all triples, the permutation of their binary representations such that the smallest element among all the triples is maximized, is searched. This prevents from adding more badly polarized bit sub-channels to the information set. Indeed, with Corollary~\ref{cor:thm_1}, the code constructed by any permutation of $m-$triples has in the same distance spectrum since the underlying information set is chosen by RM rule.
	\item Algorithm~1 of \cite{ICC_row_merging} is applied to obtain the pairs $(t,v)$, where $t\in \mathcal{N}_{\ell+1}$, $v\in \mathcal{N}_{\ell},$ $v>t$, $\ell=i_1(\mathsf{b}_i)$, to decrease the number of minimum weight codewords. 
\end{itemize}
\begin{remark}\label{rmk:RMRK1}
	Even though we have verified experimentally that the application of the third step does not decrease the minimum distance, an explicit proof of this evidence is complex and remains to be done.
\end{remark}

The pre-transformation matrix is constructed by adding the smallest index of each of $m-$triple to the information set and the other two indices are considered as dynamic frozen bits. For any obtained pair $(t,v), \, v$ is considered as the dynamic frozen bit. The pre-transformation matrix $\mathsf{T}$, is such that
\begin{align}
	T_{a,a}=T_{a,b}=T_{a,c}=T_{t,v}=1
\end{align}
where $a\in \mathcal{N}$ is the minimum of the triples, and $v\in \mathcal{N}_{\ell}$ is the associated index to any $t\in \mathcal{N}_{\ell+1}$ by the application of Algorithm~1 of~\cite{ICC_row_merging} to obtain pairs instead of triples. 

Table \ref{Fig:Table_main} summarizes the characteristics of the codes that we can construct with our method. Each entry of the table is a triple $(m,k,d)$ where $m$ is the number of information bits that can be added on the initial $k$ information bits and $d$ is the minimum distance. The code parameters are given according to two other parameters, $(n,r)$ representing the recursive number and the order of the RM polar codes, respectively. Moreover,'$-$' means that Theorem~\ref{thm:Last_ensemble} cannot be applied for the corresponding RM$(n,r)$. For $n=6$ and $n=7$, we have codewords of length $64$ and $128$ respectively. It can be seen that, for instance, for a block length 128 and $r=3$, the code rate can be extended from $23/64$ to $24/64$ while keeping the same minimum distance, which is an interesting improvement at this short block length.

\section{Simulation Results} \label{sec:simul}

We numerically compare in Figure~\ref{Fig:Main_comp} our proposed design (PD) with PAC codes and the saddle-point approximation of the MC (SMC) bound \cite{Philippe} for the binary input additive white Gaussian noise channel. Our construction for the code $(128,66)$ is obtained by first adding two extra bits to the $(128,64)$ polar-like code with RM information set and then by applying \cite[Algortihm 1]{ICC_row_merging} to obtain $(t,v)$ pairs such that $i_1(\mathsf{b}_t\bar{\cap}\mathsf{b}_v)=1$. Similarly, the code $(128,100)$ is obtained by first adding one extra bit to the polar-like code $(128,99)$ with RM information set and then by applying \cite[Algortihm 1]{ICC_row_merging} to obtain $(t,v)$ pairs such that $i_1(\mathsf{b}_t\bar{\cap}\mathsf{b}_v)=0$.

For PAC codes, the additional information indices are chosen as the most reliable bit subchannel indices from the set $\mathcal{N}_{\ell}$, which are the highest indices due to partial ordering \cite{Partial_order}. We optimize the polynomial of the convolutional code with memory length $7$ to minimize the number of minimum weight codewords. We implemented the algorithm \cite{Tse} with a large list size, i.e. $ 5 \cdot 10^4 $, and we choose the one that leads to the minimal number of second minimum weight codewords since the number of minimum weight codewords does not change for a few increment of the information length.
\begin{figure}
	\centering
	\begin{tikzpicture}
		\begin{semilogyaxis}[%
			width=1\columnwidth,
			height=10\baselineskip,
			xmin=4.2,
			xmax=7.1,
			xmajorgrids,
			xlabel={$\mathsf{E_b}/\mathsf{N_0}$ in dB},
			ymin=0.0002,
			ymax=0.15,
			ymajorgrids,
			ylabel={Frame Error Rate},
			legend style={ legend columns=3,at={(0.412,1)},anchor=south,draw=none,fill=none,legend cell align=left, font=\scriptsize}
			]
			\addplot[color=black,mark=o,every mark/.append style=solid,densely dashed] coordinates {
				(1.34+2.943,1144/20000)
				(1.8176+2.943,459/33563)
				(2.2437+2.943,255/100000)
				(2.6764+2.943,213/734138)
			};
			\addlegendentry{$k=66$, SMC} 
			
			
			\addplot[color=blue,mark=o] coordinates {
				(1.5+2.943,1144/20000)
				(2+2.943,459/33563)
				(2.5+2.943,255/100000)
				(3+2.943,213/734138)
			};
			\addlegendentry{$k=66$, PD}
			
			\addplot[color=green!80!black!100!,mark=o] coordinates {
				(1.5+2.943,657/10000)
				(2+2.943,235/11000)
				(2.5+2.943,295/49704)
				(3+2.943,218/119749)
			};
			\addlegendentry{$k=66$, PAC}
			
			\addplot[color=black,mark=triangle,every mark/.append style=solid,densely dashed] coordinates {
				(4.4694+1.0721,1205/15000)
				(4.8058+1.0721,422/14287)
				(5.26+1.0721,356/60000)
				(5.6575+1.0721,249/360698)
			};
			\addlegendentry{$k=100$, SMC}
			
			
			\addplot[color=blue,mark=triangle] coordinates {
				(4.5+1.0721,1205/15000)
				(5+1.0721,422/14287)
				(5.5+1.0721,356/60000)
				(6+1.0721,249/360698)
			};
			\addlegendentry{$k=100$, PD}
			
			\addplot[color=green,mark=triangle] coordinates {
				(4.5+1.0721,1081/10000)
				(5+1.0721,447/12551)
				(5.5+1.0721,433/50000)
				(6+1.0721,114/70398)
			};
			\addlegendentry{$k=100$, PAC}
			
		\end{semilogyaxis}%
	\end{tikzpicture}
	\caption{FER of our proposed scheme (PD), compared to SMC and PAC codes, $N=128$, $k\in\{66;100\}$.}
	\label{Fig:Main_comp}
\end{figure}
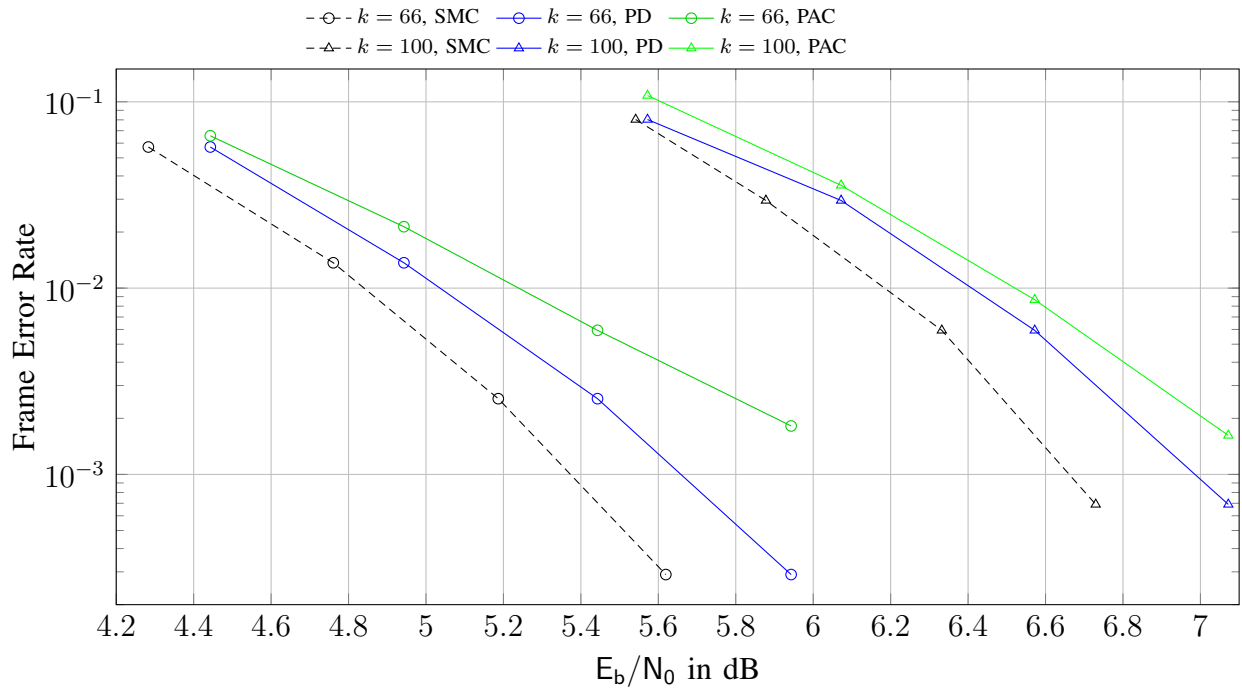

Figure \ref{Fig:Main_comp} plots the frame error rate (FER) versus $E_b / N_0$ for a code length of $N=128$ in an additive white Gaussian noise channel and two information length, i.e. $k=66$ and $k=100$. Our proposed design (PD) outperforms PAC codes for the entire range of $E_b / N_0$, since, at short block lengths, the minimum distance plays an important role in the SCL decoding with large list sizes. In particular, for $(128,66)$ code, while our design performs within $0.25$ dB of SMC bound at FER of $3.10^{-4}$, PAC code achieves the same performance with $0.4$ dB of additional power. For $(128,100)$ codes, our design outperforms PAC code of about $0.2$ dB at FER of $10^{-3}$.

\section{Concluding Remarks} \label{sec:conclusion}

In this work, we proposed a method to increase the information length of a polar-like code while keeping the same minimum distance with the underlying RM code. Our findings allow to reduce the number of minimum weight codewords of polar-like codes that perform closer to the MC bound than PAC codes with the same system parameters. We believe that this work may lead to a new method for code design, particularly at short block lengths, with interesting performance. The extension of this work to moderate block lengths is under investigation.

\bibliographystyle{IEEEtran}
\bibliography{references}


\section{Appendices}

\subsection{Preliminary Theorems}
In this section, we provide preliminary results that are useful in the proof of Theorem~\ref{thm:main_min_dis_inc_three_rows} and Theorem~\ref{thm:Last_ensemble}.


\begin{theorem}\label{thm:main_min_dis_inc}
	For any pairs 
	$(i,j)\in \mathcal{N}_{\ell}$, $\mathcal{N}_{\ell}:=\{t: i_1(\mathsf{b}_t)=\ell, t\in \mathcal{N}\}$, such that $i_1(\mathsf{b}_i\bar{\cap}\mathsf{b}_j)=0$, the combination of $\mathsf{g}_i\oplus\mathsf{g}_j$ with higher hamming weight rows of polar encoding matrix is lower bounded by Hamming weight of $\mathsf{g}_i\oplus\mathsf{g}_j$:
	\begin{align}\label{eq:eqt_thm4}
		i_1(\mathsf{g}_i\oplus\mathsf{g}_j\oplus \mathsf{g}_{\mathcal{T}})&\geq i_1(\mathsf{g}_i\oplus\mathsf{g}_j)\overset{(a)}{=}2^{\ell+1}-2
	\end{align}
	where $\mathcal{T}\subseteq \cup_{p=\ell+1}^{n}\mathcal{N}_p$ and (a) is by Theorem~\ref{th:Hamm_of_sum0000}.
\end{theorem}

\begin{proof} For $\ell=1$, the statement of the theorem turns to be trivial due to the fact that $i_1(\mathsf{g}_{\{i,j,\mathcal{T}\}})\geq 2^{\min_{k\in \{i,j,\mathcal{T}\}}i_1(\mathsf{b}_k)}=i_1(\mathsf{g}_i\oplus\mathsf{g}_j)$ for $\ell=1$. Hence, in the following, the proof is conducted for $\ell\geq 2$.
	
	Due to the symmetry imposed by Corollary~\ref{cor:thm_1}, without losing generality we assume that $\mathcal{P}_1(\mathsf{b}_i)=[0,\ell-1]$ and $\mathcal{P}_1(\mathsf{b}_j)=[\ell,2\ell-1]$. In the following we use this assumption for the ease of presentation.
	
	We can divide the index set into two subsets such that
	\begin{align}\label{eq:in_thm_subset1}
		\hat{\mathcal{T}}:=\{k:\mathcal{P}_1(\mathsf{b}_k)\cap[2\cdot\ell,n)= \emptyset, k\in \mathcal{T} \}
	\end{align}
	and $\tilde{\mathcal{T}}=\mathcal{T}\setminus\hat{\mathcal{T}}$.
	Then, by Theorem~\ref{thm:main_row_merg0}
	\begin{align}\label{eq:outside_of_scope}
		i_1(\mathsf{g}_i\oplus\mathsf{g}_j\bigoplus_{k\in \hat{\mathcal{T}}}\mathsf{g}_k\bigoplus_{t\in \tilde{\mathcal{T}}}\mathsf{g}_t) &\geq \max_{p_0\in [2\cdot \ell,n)}i_1(\mathsf{g}_i\oplus\mathsf{g}_j\bigoplus_{k\in \hat{\mathcal{T}}}\mathsf{g}_k\bigoplus_{t\in \tilde{\mathcal{T}}}\mathsf{g}_t \mathbb{I}\{b_{t,p_0}=0\})\nonumber\\
		&\geq \max_{p_1\in [2\cdot \ell,n)\setminus p_0} i_1(\mathsf{g}_i\oplus \mathsf{g}_j \bigoplus_{k\in \hat{\mathcal{T}}}\mathsf{g}_k\bigoplus_{t\in \tilde{\mathcal{T}}}\mathsf{g}_t \mathbb{I}\{b_{t,p_0}={b_{t,p_1}\hspace{-0.2cm}=\hspace{-0.07cm}0}\})\nonumber \\
		&\hspace{0.2cm} \vdots \nonumber \\
		&\overset{(a)}{\geq} i_1(\mathsf{g}_i\hspace{-0.07cm}\oplus \mathsf{g}_j \hspace{-0.07cm}\bigoplus_{k\in \hat{\mathcal{T}}}\mathsf{g}_k\bigoplus_{t\in \tilde{\mathcal{T}}}\mathsf{g}_t \mathbb{I}\{b_{t,p_0}\hspace{-0.2cm}=\hspace{-0.08cm}{b_{t,p_1}\hspace{-0.2cm}=\hspace{-0.07cm}\cdots\hspace{-0.07cm} = \hspace{-0.07cm}b_{t,p_{n-2 \ell -1}}\hspace{-0.2cm}=\hspace{-0.07cm}0}\})\nonumber \\
		&\overset{(b)}{=}i_1(\mathsf{g}_i\oplus \mathsf{g}_j \bigoplus_{k\in \hat{\mathcal{T}}}\mathsf{g}_k )
	\end{align}
	where (a) comes from the repeated application of Theorem~\ref{thm:main_row_merg0} in $[2\ell,n)$, and (b) comes from \eqref{eq:in_thm_subset1}, which implies for any $t \in \tilde{\mathcal{T}}$, $\mathcal{P}_1(\mathsf{b}_t)\cap [2\cdot\ell,n)\neq \emptyset$. 
	Let us continue for $\hat{\mathcal{T}}$ which is divided into two subsets such that
	\begin{align}\label{eq:t_00_new}
		\hat{\mathcal{T}}^{0}\hspace{-0.13cm}:=\hspace{-0.1cm}\{k\hspace{-0.1cm}:\hspace{-0.1cm} \mathcal{P}_1(\mathsf{b}_k)\hspace{-0.1cm}\nsupseteq \hspace{-0.1cm} \mathcal{P}_1(\mathsf{b}_i),\mathcal{P}_1(\mathsf{b}_k)\hspace{-0.1cm}\nsupseteq \hspace{-0.1cm}\mathcal{P}_1(\mathsf{b}_j),k\hspace{-0.1cm}\in\hspace{-0.1cm}\hat{\mathcal{T}} \}
	\end{align}
	and
	\begin{align}\label{eq:t_01_new}
		\hspace{-0.2cm}\hat{\mathcal{T}}^{1}\hspace{-0.13cm}:=\hspace{-0.1cm}\{k\hspace{-0.1cm}: \hspace{-0.1cm}\mathcal{P}_1(\mathsf{b}_k)\hspace{-0.1cm}\supset\hspace{-0.1cm} \mathcal{P}_1(\mathsf{b}_i)\text{ or }\mathcal{P}_1(\mathsf{b}_k)\hspace{-0.1cm}\supset\hspace{-0.1cm} \mathcal{P}_1(\mathsf{b}_j),\hspace{-0.05cm}k\hspace{-0.1cm}\in\hspace{-0.1cm}\hat{\mathcal{T}} \}
	\end{align}

	Two cases have to be investigated, i.e. when $\hat{\mathcal{T}}^{0}$  is not empty and $\hat{\mathcal{T}}^{0}$ is empty.

	\paragraph*{Case~1, $\hat{\mathcal{T}}^{0} \neq \emptyset$} The elements of this set are partitioned with respect to the Hamming weights of their binary representations:
	
	\begin{align}
		\hat{\mathcal{T}}^{0}=\bigcup_{d=1}^{n-\ell} 	\hat{\mathcal{T}}_{\ell+d}^{0}, \; \; \; \;   \hat{\mathcal{T}}_{\ell+d}^{0}:= \{k: i_1(\mathsf{b}_k)=\ell+d, k\in \hat{\mathcal{T}}^{0} \}
	\end{align}

	Now, let us assume that $d^{*}$ be the minimum number such that $\hat{\mathcal{T}}_{\ell+d}^{0}$ is not empty, i.e., $d^{*}=\min \{d: \hat{\mathcal{T}}_{\ell+d}^{0}\neq \emptyset, n-\ell \geq d\geq 1\}$. Note that for any $t\in \hat{\mathcal{T}}^{0}_{\ell+d^{*}}$, $\exists p_0\in \mathcal{P}_0(\mathsf{b}_i)\cap \mathcal{P}_0(\mathsf{b}_t)\cap \mathcal{P}_1(\mathsf{b}_j) \text{ and } p_1\in \mathcal{P}_1(\mathsf{b}_i)\cap\mathcal{P}_0(\mathsf{b}_j) \cap \mathcal{P}_0(\mathsf{b}_t)$ such that
	\begin{align}\label{eq:eqCase_1_1_3}
		i_1(&\mathsf{g}_i\oplus \mathsf{g}_j\oplus \mathsf{g}_{\hat{\mathcal{T}}^{0}}\oplus \mathsf{g}_{\hat{\mathcal{T}}^{1}}) \nonumber \\
		 &\overset{(a)}{\geq} i_1(\mathsf{g}_i\mathbb{I}\{b_{i,p_0}=0\} \oplus\mathsf{g}_j\mathbb{I}\{b_{j,p_0}=0\}\oplus \mathsf{g}_t\mathbb{I}\{b_{t,p_0}=0\}\bigoplus_{k\in\hat{\mathcal{T}}^{0}\setminus t}\mathsf{g}_k\mathbb{I}\{b_{k,p_0}=0\}\bigoplus_{k\in\hat{\mathcal{T}}^{1}}\mathsf{g}_k\mathbb{I}\{b_{k,p_0}\hspace{-0.1cm}=0\})\nonumber \\
		& \overset{(b)}{\geq} i_1(\mathsf{g}_i\oplus \mathsf{g}_t\bigoplus_{k\in\hat{\mathcal{T}}^{0}\setminus t}\mathsf{g}_k\mathbb{I}\{b_{k,p_0}=0\}\bigoplus_{k\in\hat{\mathcal{T}}^{1}}\mathsf{g}_k\mathbb{I}\{b_{k,p_0}=0\})\nonumber \\
		& & \nonumber \\
		&\overset{(c)}{\geq} i_1(\mathsf{g}_i\mathbb{I}\{b_{i,p_0}=b_{i,p_1}=0\}\mathsf{g}_t\mathbb{I}\{b_{t,p_0}=b_{t,p_1}=0\}\bigoplus_{k\in\hat{\mathcal{T}}^{0}\setminus t}\mathsf{g}_k\mathbb{I}\{b_{k,p_0}=b_{k,p_1}=0\}\nonumber \\ & \hspace{10cm}\bigoplus_{k\in\hat{\mathcal{T}}^{0}}\mathsf{g}_k\mathbb{I}\{b_{k,p_0}=b_{k,p_1}=0\})\nonumber \\
		&=i_1(\bigoplus_{k\in\hat{\mathcal{T}}^{0}}\mathsf{g}_k\mathbb{I}\{b_{k,p_0}=b_{k,p_1}=0\})\overset{(d)}{\geq} 2^{\ell+d^{*}}
	\end{align}
	where (a) and (c) come from Theorem~\ref{thm:main_row_merg0} and from the conditions imposed by \eqref{eq:t_00_new} and \eqref{eq:t_01_new}, and (b) comes from the fact $b_{j,p_0} = 1$ by hypothesis.  Moreover (d) is due to the fact that $i_1(\mathsf{g}_{\mathcal{S}})\geq \min_{j\in \mathcal{S}}2^{i_1(\mathsf{b}_j)}$, i.e., by \eqref{eq:low_bound_d}.
		That means if $\hat{\mathcal{T}}^{0}$ is not empty, the lower bound is satisfied whatever $\hat{\mathcal{T}}^{1}$ is empty or not. 
	
 	\paragraph*{Case~2, $\hat{\mathcal{T}}^{0} = \emptyset$} The elements of $\hat{\mathcal{T}}^{1}$ are partitioned such as 
	
		\begin{align}\label{eq:cas2_01000}
		\hat{\mathcal{T}}_{i}^{1}=\begin{cases}
			\{k:\mathcal{P}_1(\mathsf{b}_k)\supset \mathcal{P}_1(\mathsf{b}_i), k\in \hat{\mathcal{T}}^{1} \}\setminus \{2^{2\cdot \ell}-1\}& \text{if } 2^{2\cdot \ell}-1\in \hat{\mathcal{T}}^{1}\\
			\{k:\mathcal{P}_1(\mathsf{b}_k)\supset \mathcal{P}_1(\mathsf{b}_i), k\in \hat{\mathcal{T}}^{1} \} & \text{ otherwise}
		\end{cases}
	\end{align}
	 and 
	
	\begin{align}\label{eq:cas2_01001}
			\hat{\mathcal{T}}_{j}^{1}= \{k:\mathcal{P}_1(\mathsf{b}_k)\supset \mathcal{P}_1(\mathsf{b}_j), k\in \hat{\mathcal{T}}^{1} \}
		\end{align}
	Then, we obtain the total Hamming weight by partitioning the indices of row vectors into $2^{\ell}$ subsets with respect to binary representation indices $[0,\ell-1]$ and using Definition \ref{def:2}:
	 			\begin{align}\label{eq:subset_proj0001}
		 				i_1(\mathsf{g}_i\oplus \mathsf{g}_j\oplus\mathsf{g}_{\hat{\mathcal{T}}_{i}^{0}}\oplus\mathsf{g}_{\hat{\mathcal{T}}_{i}^{1}})&=i_1(\mathsf{g}_i\hspace{-0.08cm} \oplus \hspace{-0.08cm} \mathsf{g}_j\hspace{-0.08cm} \oplus\mathsf{g}_{\hat{\mathcal{T}}_{i}^{0}}\oplus\mathsf{g}_{\hat{\mathcal{T}}_{i}^{1}}|\mathcal{M}^{c}_{\ell-1}\hspace{-0.08cm}\cap\hspace{-0.08cm} \mathcal{M}^{c}_{\ell-2}\hspace{-0.08cm}\cap\hspace{-0.08cm}\cdots \hspace{-0.08cm}\cap\hspace{-0.08cm} \mathcal{M}^{c}_{1}\hspace{-0.08cm}\cap\mathcal{M}^{c}_{0})\nonumber \\ &+	i_1(\mathsf{g}_i\hspace{-0.08cm} \oplus \hspace{-0.08cm} \mathsf{g}_j\hspace{-0.08cm}\oplus\mathsf{g}_{\hat{\mathcal{T}}_{i}^{0}}\oplus\mathsf{g}_{\hat{\mathcal{T}}_{i}^{1}}|\mathcal{M}^{c}_{\ell-1}\hspace{-0.08cm} \cap \hspace{-0.08cm} \mathcal{M}^{c}_{\ell-2}\hspace{-0.08cm} \cap\hspace{-0.08cm} \cdots \hspace{-0.08cm}\cap \hspace{-0.08cm} \mathcal{M}^{c}_{1}\cap\hspace{-0.08cm}\mathcal{M}_{0}) \cdots \nonumber \\ \cdots
		 				&+ i_1(\mathsf{g}_i\hspace{-0.08cm}\oplus\hspace{-0.08cm} \mathsf{g}_j\hspace{-0.08cm}\oplus\mathsf{g}_{\hat{\mathcal{T}}_{i}^{0}}\oplus\mathsf{g}_{\hat{\mathcal{T}}_{i}^{1}}|\mathcal{M}_{\ell-1}\hspace{-0.08cm}\cap\hspace{-0.08cm} \mathcal{M}_{\ell-2}\hspace{-0.08cm}\cap\hspace{-0.08cm}\cdots\hspace{-0.08cm} \cap\hspace{-0.08cm} \mathcal{M}_{1}\cap\hspace{-0.08cm} \mathcal{M}_{0})
		 			\end{align} 
	 			By Definition~\ref{def:2}, we can write \eqref{eq:subset_proj0001} in a more compact form by denoting $\mathcal{M}^{c}_{p}$ with $0$ and $\mathcal{M}_{p}$ with $1$ for $p\in [0,\ell-1]$, then

	 			\begin{align}\label{eq:Case2_new_proof_000}
		 					i_1(\mathsf{g}_i\oplus \mathsf{g}_j\oplus\mathsf{g}_{\hat{\mathcal{T}}_{i}^{1}}\oplus\mathsf{g}_{\hat{\mathcal{T}}_{j}^{1}})&=\hspace{-0.1cm}\sum_{k=0}^{2^{\ell}-1}\hspace{-0.1cm}i_1(\mathsf{g}^{[0, \ell-1]}_i\mathbb{I}\{\mathcal{P}_1(\mathsf{b}_i)\hspace{-0.1cm}\supset\hspace{-0.1cm} \mathcal{P}_1(\mathsf{b}_k)\}\hspace{-0.07cm}\oplus\hspace{-0.07cm}\mathsf{g}^{[0, \ell-1]}_j\mathbb{I}\{\mathcal{P}_1(\mathsf{b}_j)\hspace{-0.1cm}\supset\hspace{-0.1cm}\mathcal{P}_1(\mathsf{b}_k)\}\nonumber \\ &\hspace{1.5cm}\bigoplus_{t\in \hat{\mathcal{T}}_{i}^{1}}\mathsf{g}^{[0, \ell-1]}_t\mathbb{I}\{\mathcal{P}_1(\mathsf{b}_t)\hspace{-0.1cm}\supset\hspace{-0.1cm}\mathcal{P}_1(\mathsf{b}_k)\}\bigoplus_{t\in \hat{\mathcal{T}}_{j}^{1}}\mathsf{g}^{[0, \ell-1]}_t\mathbb{I}\{\mathcal{P}_1(\mathsf{b}_t)\hspace{-0.1cm}\supset\hspace{-0.1cm}\mathcal{P}_1(\mathsf{b}_k)\}))\nonumber \\
&\overset{(a)}{=}\sum_{k=0}^{2^{\ell}-1} i_1(\mathsf{g}^{[0, \ell-1]}_{\{i,\hat{\mathcal{T}}_{i}^{1}\}}\bigoplus_{t\in \{j,\hat{\mathcal{T}}_{j}^{1}\}}\mathsf{g}^{[0, \ell-1]}_t\mathbb{I}\{\mathcal{P}_1(\mathsf{b}_t)\hspace{-0.1cm}\supset\hspace{-0.1cm}\mathcal{P}_1(\mathsf{b}_k)\})
\end{align}
\noindent where (a) is by Definition~\ref{def:2} and conditions imposed by  \eqref{eq:cas2_01000} and \eqref{eq:cas2_01001}, i.e., $\mathcal{P}_1(\mathsf{b}_t)\hspace{-0.1cm}\supset\hspace{-0.1cm}\mathcal{P}_1(\mathsf{b}_k)$ for any $t \in \{i,\mathcal{T}_{i}^{1} \}$.

Note that since $i_1(\mathsf{g}_{t}^{[0,\ell-1]})=2^{\ell}, \; t\in \{j,\hat{\mathcal{T}}_{j}^{1}\}$, we have
\begin{align}\label{eq:Case2_new_proof_001}
	i_1(\bigoplus_{t\in \{j,\hat{\mathcal{T}}_{j}^{1}\}}\mathsf{g}^{[0, \ell-1]}_t\mathbb{I}\{\mathcal{P}_1(\mathsf{b}_t)\hspace{-0.1cm}\supset\hspace{-0.1cm}\mathcal{P}_1(\mathsf{b}_k)\})=\begin{cases}
		2^{\ell}& \text{if } |\{ t: \mathcal{P}_1(\mathsf{b}_t)\hspace{-0.1cm}\supset\hspace{-0.1cm}\mathcal{P}_1(\mathsf{b}_k), t\in \{j,\hat{\mathcal{T}}_{j}^{1} \}| \text{  is odd,}\\
		0 & \text{ otherwise}
		\end{cases}
\end{align}
and 
\begin{align}\label{eq:Case2_new_proof_010}
1	\overset{(a)}{\leq} i_1(\mathsf{g}^{[0, \ell-1]}_{\{i,\hat{\mathcal{T}}_{i}^{1}\}})\overset{(b)}{\leq} 2^{\ell}-1 
\end{align}	
where (a) is by  \eqref{eq:low_bound_d} and (b) is by  the property of polar encoding matrix that is $\mathsf{g}_{t,2^{n}-1}=0,  \; 0 \leq t\leq 2^n-2$ and $\mathsf{g}_{t,2^{n}-1}=1,  \;  \text{ if } t= 2^n-1$.  This can be deducted once it is recognized that, by \eqref{eq:subset_proj0}, each  $\mathsf{g}_{t}^{[0,\ell-1]}, \; t \in \{i,\hat{\mathcal{T}}_{i}^{1} \}$ 
is a polar matrix row and $\mathsf{g}_{t,2^{\ell}-1}^{[0,\ell-1]}=0, \; t \in \{i,\hat{\mathcal{T}}_{i}^{1} \} $ since $\mathcal{P}_1(\mathsf{b}_t)\hspace{-0.1cm}\nsupseteq[\ell, 2\cdot \ell-1] $ by \eqref{eq:cas2_01000}, which imposes (b).

By \eqref{eq:Case2_new_proof_001}, we can write \eqref{eq:Case2_new_proof_000} as 
\begin{align}\label{eq:Case2_new_proof_011}
	i_1(\mathsf{g}_i\oplus \mathsf{g}_j\oplus\mathsf{g}_{\hat{\mathcal{T}}_{i}^{1}}\oplus\mathsf{g}_{\hat{\mathcal{T}}_{j}^{1}})=\alpha \cdot i_1(\mathsf{g}^{[0, \ell-1]}_{\{i,\hat{\mathcal{T}}_{i}^{1}\}})+(2^{\ell}-\alpha)\cdot(2^{\ell}-i_1(\mathsf{g}^{[0, \ell-1]}_{\{i,\hat{\mathcal{T}}_{i}^{1}\}})).
\end{align}
Note that the $\alpha$ is such that
\begin{align}\label{eq:Case2_new_proof_110}
    1\overset{(a)}{\leq} \alpha \overset{(b)}{\leq} 2^{\ell}-1	
\end{align}
where (a)  is due to the fact that, when $k=2^{\ell}-1$ and by \eqref{eq:subset_proj1}, $\mathsf{g}^{[0, \ell-1]}_{\{i,\hat{\mathcal{T}}_{i}^{1}\}}$ is the only term in \eqref{eq:Case2_new_proof_000}, i.e.
\begin{align}
	i_1(\mathsf{g}^{[0, \ell-1]}_{\{i,\hat{\mathcal{T}}_{i}^{1}\}}\bigoplus_{t\in \{j,\mathcal{T}_{j}^{1}\}} \mathsf{g}_t\mathbb{I}\{\mathcal{P}_1(\mathsf{b}_t)\hspace{-0.1cm}\supset\hspace{-0.1cm}\mathcal{P}_1(\mathsf{b}_{2^{\ell}-1})\} )= i_1(\mathsf{g}^{[0, \ell-1]}_{\{i,\hat{\mathcal{T}}_{i}^{1}\}})
\end{align} 
and (b) comes from the fact that $\mathsf{g}_{t^{*}}, \;  t^{*}=\max\{\{j,\mathcal{T}_{j}^{1}\}\}$ is the only element from $\{j,\mathcal{T}_{j}^{1}\}$ such that \eqref{eq:Case2_new_proof_000} simplifies to
\begin{align}
    	i_1(\mathsf{g}^{[0, \ell-1]}_{\{i,\hat{\mathcal{T}}_{i}^{1}\}}\bigoplus_{t\in \{j,\mathcal{T}_{j}^{1}\}} \mathsf{g}^{[0,\ell-1]}_t\mathbb{I}\{\mathcal{P}_1(\mathsf{b}_t)\hspace{-0.1cm}\supset\hspace{-0.1cm}\mathcal{P}_1(\mathsf{b}_{ t^{*}})\} )= i_1(\mathsf{g}^{[0, \ell-1]}_{\{i,\hat{\mathcal{T}}_{i}^{1}\}}\oplus \mathsf{g}^{[0,\ell-1]}_{ t^{*}})=2^{\ell}-i_1(\mathsf{g}^{[0, \ell-1]}_{\{i,\hat{\mathcal{T}}_{i}^{1}\}})
\end{align}
by using \eqref{eq:subset_proj1}.
When \eqref{eq:Case2_new_proof_010} and \eqref{eq:Case2_new_proof_110} are satisfied, \eqref{eq:Case2_new_proof_011} cannot be less than $2^{\ell+1}-2$.
This ends the proof for Theorem~\ref{thm:main_min_dis_inc}.
\end{proof}

\begin{theorem}\label{thm:main_min_dis_inc_with_intrsctn_two_rows}
	For any pairs $(i,j)\in \mathcal{N}_{\ell}, \ell\geq 2$, the combination of $\mathsf{g}_i\oplus\mathsf{g}_j$ with higher Hamming weight rows of polar encoding matrix is lower bounded by Hamming weight of $\mathsf{g}_i\oplus\mathsf{g}_j$
	\begin{align}\label{eq:eqt_thm4_two_rows}
		i_1(\mathsf{g}_i\oplus\mathsf{g}_j\oplus \mathsf{g}_{\mathcal{T}})&\geq i_1(\mathsf{g}_i\oplus\mathsf{g}_j)
	\end{align}
	where $\mathcal{T}\subseteq \cup_{p=\ell+1}^{n}\mathcal{N}_p$.
\end{theorem}

\begin{proof}
	We can divide the index set into two subsets such that
	\begin{align}\label{eq:in_thm_subset0100}
		\tilde{\mathcal{T}}:=\{k:\mathcal{P}_1(\mathsf{b}_k)\cap\mathcal{P}_0(\mathsf{b}_i\bar{\cup}\mathsf{b}_j)\neq \emptyset, k\in \mathcal{T} \}
	\end{align}
	and $\hat{\mathcal{T}}=\mathcal{T}\setminus \tilde{\mathcal{T}}$.
	
	Applying the same reasoning as in Theorem~\ref{thm:main_min_dis_inc} in \eqref{eq:outside_of_scope}, the proof is conducted for $\hat{\mathcal{T}}$. 
	
	When $i_1(\mathsf{b}_i\bar{\cap}\mathsf{b}_j)=\ell-1$, the proof is trivial since 
	\begin{align}
		i_1(\mathsf{g}_{\hat{\mathcal{T}}\cup \{i,j\}})\geq \min_{k\in \hat{\mathcal{T}}\cup \{i,j\}}2^{i_1(\mathsf{b}_k)}=2^{i_1(\mathsf{b}_i)}=2^{\ell}
	\end{align}
	and 	$i_1(\mathsf{g}_i\oplus \mathsf{g}_j)=2^{\ell}$ by \cite[Theorem~2]{ICC_row_merging} if $i_1(\mathsf{b}_i\bar{\cap}\mathsf{b}_j)=\ell-1$.
	
	Now assume that $i_1(\mathsf{b}_i\bar{\cap}\mathsf{b}_j)=|\mathcal{P}_1(\mathsf{b}_i\bar{\cap}\mathsf{b}_j)|\leq \ell-2$ and let $\mathcal{W}=\mathcal{P}_1(\mathsf{b}_i\bar{\cap}\mathsf{b}_j)$, which is the index set of common one bit positions of $\mathsf{b}_i$ and $\mathsf{b}_j$. Then, by partitioning the row indices of the polar encoding matrix, we obtain the following expression
	\begin{align}
		i_1(\mathsf{g}_{\{i,j\}}  \oplus \mathsf{g}_{\hat{\mathcal{T}}}) &=i_1(\mathsf{g}_{\{i,j\}} \hspace{-0.12cm} \oplus \hspace{-0.08cm} \mathsf{g}_{\hat{\mathcal{T}}}| \mathcal{M}^{c}_{\mathcal{W}(|\mathcal{W}|-1)}\hspace{-0.08cm}\cap\hspace{-0.08cm} \mathcal{M}^{c}_{\mathcal{W}(|\mathcal{W}|-2)}\hspace{-0.08cm}\cap\hspace{-0.08cm} \cdots \hspace{-0.08cm} \cap \hspace{-0.08cm} \mathcal{M}^{c}_{ \mathcal{W}(0)})  \\ &+i_1(\mathsf{g}_{\{i,j\}} \hspace{-0.12cm} \oplus \hspace{-0.08cm} \mathsf{g}_{\hat{\mathcal{T}}}| \mathcal{M}^{c}_{\mathcal{W}(|\mathcal{W}|-1)}\hspace{-0.08cm} \cap \hspace{-0.08cm} \mathcal{M}^{c}_{\mathcal{W}(|\mathcal{W}|-2)}\hspace{-0.08cm} \cap \hspace{-0.08cm} \cdots \hspace{-0.08cm}\cap \hspace{-0.08cm} \mathcal{M}_{ \mathcal{W}(0)}) \nonumber \\ & \hspace{1.55cm} \vdots \nonumber \\ &+i_1(\mathsf{g}_{\{i,j\}} \hspace{-0.12cm} \oplus \hspace{-0.08cm} \mathsf{g}_{\hat{\mathcal{T}}}|\mathcal{M}_{\mathcal{W}(|\mathcal{W}|-1)} \hspace{-0.08cm} \cap \hspace{-0.08cm} \mathcal{M}_{\mathcal{W}(|\mathcal{W}|-2)} \hspace{-0.08cm} \cap \hspace{-0.08cm} \cdots \hspace{-0.08cm} \cap \hspace{-0.08cm} \mathcal{M}_{ \mathcal{W}(0)})\nonumber \\
		& \overset{(a)}{=} i_1(\mathsf{g}^{\mathcal{W}}_{\{i,j\}} \bigoplus_{k\in \hat{\mathcal{T}}} \mathsf{g}^{\mathcal{W}}_{k})+ i_1(\mathsf{g}^{\mathcal{W}}_{\{i,j\}}\bigoplus_{k\in \hat{\mathcal{T}}} \mathsf{g}^{\mathcal{W}}_{k}\mathbb{I}\{b_{k,\mathcal{W}(0)}=1\}) \nonumber \\ & \hspace{1.5cm}\vdots \nonumber \\ &+ i_1(\mathsf{g}^{\mathcal{W}}_{\{i,j\}} \bigoplus_{k\in \hat{\mathcal{T}}} \mathsf{g}^{\mathcal{W}}_{k}\mathbb{I}\{b_{k,\mathcal{W}(|\mathcal{W}|-1)}=b_{k,\mathcal{W}(|\mathcal{W}|-2)}=\cdots=b_{k,\mathcal{W}(0)}=1\})\notag
	\end{align}
	where (a) is due to \eqref{eq:subset_proj0} and \eqref{eq:subset_proj1}. By Theorem~\ref{thm:main_min_dis_inc}, each term of the partition is greater than $2^{\ell-|\mathcal{W}|+1}-2$ since $|\mathcal{P}_1(\mathsf{b}_k)\cap \mathcal{P}_1(\mathsf{b}_i\oplus \mathsf{b}_j)|>\ell-|\mathcal{W}|$ for any $k\in \hat{\mathcal{T}}$ by assumption. Then, 
	\begin{align}
		i_1(\mathsf{g}_i\oplus \mathsf{g}_j \oplus \mathsf{g}_{\hat{\mathcal{T}}})&\geq2^{i_1(\mathsf{b}_i\bar{\cap} \mathsf{b}_j)}\cdot (2^{\ell-i_1(\mathsf{b}_i\bar{\cap} \mathsf{b}_j)+1}-2)=2^{\ell+1}-2^{i_1(\mathsf{b}_i\bar{\cap} \mathsf{b}_j)+1}\nonumber \\
		& \overset{(a)}{=} 	i_1(\mathsf{g}_i\oplus \mathsf{g}_j)
	\end{align}
	where (a) is due to Theorem~\ref{th:Hamm_of_sum0000}.
\end{proof}

\subsection{Proof of Theorem~\ref{thm:main_min_dis_inc_three_rows}}\label{app:thm3}

The proof relies on Theorem~\ref{thm:main_min_dis_inc}, given at the beginning of the appendices.  Theorem~\ref{thm:main_min_dis_inc} states that, for any given pair of rows with the same Hamming weight and no intersection in their binary representations, the Hamming weights of combination of the given pair with any subset of rows with higher Hamming weights is lower bounded by the Hamming weight of the given pair.

Since 
\begin{align}
d(\bar{\mathbf{C}})&=\min\{d(\mathbf{C}), \min_{\mathcal{T}\subseteq \mathcal{A}}	i_1(\mathsf{g}_{\{i,j,k\}}\oplus \mathsf{g}_{\mathcal{T}})\} %
\end{align}
it is sufficient to prove the following statement 
\begin{align}
i_1(\mathsf{g}_{\{i,j,k\}}\oplus \mathsf{g}_{\mathcal{T}})&\geq 2^{\ell+1}, \; \; \forall \mathcal{T}\subseteq \mathcal{A}
\end{align}
to prove the theorem.


We can divide the index set into two subsets such that
\begin{align}\label{eq:in_thm_subset010012}
\tilde{\mathcal{T}}:=\{t:\mathcal{P}_1(\mathsf{b}_t)\cap\mathcal{P}_0(\mathsf{b}_i\bar{\cup}\mathsf{b}_j\bar{\cup}\mathsf{b}_k)\neq \emptyset, t\in \mathcal{T} \}
\end{align}
and $\hat{\mathcal{T}}=\mathcal{T}\setminus \tilde{\mathcal{T}}$.
Then by Theorem~\ref{thm:main_row_merg0}
\begin{align}
i_1(\mathsf{g}_{\{i,j,k\}}\oplus\mathsf{g}_{\hat{\mathcal{T}}}\oplus\mathsf{g}_{ \tilde{\mathcal{T}}})
&\geq \max_{p_0\in \mathcal{P}_0(\mathsf{b}_i\bar{\cup}\mathsf{b}_j\bar{\cup}\mathsf{b}_k))}i_1(\mathsf{g}_{\{i,j\}}\oplus \mathsf{g}_{\hat{\mathcal{T}}} \bigoplus_{t\in \tilde{\mathcal{T}}}\mathsf{g}_t \mathbb{I}\{b_{t,p_0}=0\})\nonumber\\
&\geq \hspace{-0.25cm} \max_{p_1\in \mathcal{P}_0(\mathsf{b}_i\bar{\cup}\mathsf{b}_j\bar{\cup}\mathsf{b}_k))\setminus \{p_0\}} i_1(\mathsf{g}_{\{i,j,k\}}\oplus \mathsf{g}_{\hat{\mathcal{T}}}\bigoplus_{t\in \tilde{\mathcal{T}}}\mathsf{g}_t \mathbb{I}\{b_{t,p_0}={b_{t,p_1}=0}\})\nonumber \\
&\hspace{0.15cm}\vdots \nonumber \\
&\geq i_1(\mathsf{g}_{\{i,j,k\}} \hspace{-0.08cm} \oplus \hspace{-0.1cm} \mathsf{g}_{\hat{\mathcal{T}}} \bigoplus_{t\in \tilde{\mathcal{T}}}\mathsf{g}_t \mathbb{I}\{b_{t,p_0}\hspace{-0.15cm}= \hspace{-0.08cm}{b_{t,p_1} \hspace{-0.1cm}= \hspace{-0.05cm}\cdots \hspace{-0.05cm} = b_{t,p_{n-2 \ell -1}}=0}\})\nonumber \\
&\overset{(a)}{=}i_1(\mathsf{g}_i\oplus \mathsf{g}_j\oplus \mathsf{g}_k \oplus \mathsf{g}_{\hat{\mathcal{T}}})
\end{align}
where $\{p_0,p_1,\cdots,p_{n-2\ell-1}\}=\mathcal{P}_0(\mathsf{b}_i\bar{\cup}\mathsf{b}_j\bar{\cup}\mathsf{b}_k)$ and (a) comes from \eqref{eq:in_thm_subset010012}, which implies for any $t\in \tilde{\mathcal{T}}$, $\mathcal{P}_1(\mathsf{b}_t)\cap\mathcal{P}_0(\mathsf{b}_i\bar{\cup}\mathsf{b}_j\bar{\cup}\mathsf{b}_k)\neq \emptyset$. This result means that the Hamming weight of $\mathsf{g}_{\{i,j,k,\mathcal{T}\}}$ is lower bounded by the Hamming weight of $\mathsf{g}_{\{i,j,k,\hat{\mathcal{T}}\}}$. Therefore, in the following, the proof is conducted for $\hat{\mathcal{T}}$. 

We divide $\hat{\mathcal{T}}$ into three subsets
\begin{align}
&\hat{\mathcal{T}}_{k,1}=\{t:\mathcal{P}_1(\mathsf{b}_t)\supset \mathcal{P}_1(\mathsf{b}_k), t\in \hat{\mathcal{T}}\} \\
&\hat{\mathcal{T}}_{k,2}=\{t:|\mathcal{P}_1(\mathsf{b}_t)\cap \mathcal{P}_1(\mathsf{b}_k)|=1, t\in \hat{\mathcal{T}}\} \\
& \tilde{\mathcal{T}}_0=\hat{\mathcal{T}}\setminus \{\hat{\mathcal{T}}_{k,1}\cup\hat{\mathcal{T}}_{k,2}\}
\end{align}

In the following, we show that the theorem holds whatever the sets $\hat{\mathcal{T}}_{k,1}$ and $\hat{\mathcal{T}}_{k,2}$ are empty or not.

\paragraph*{Case~1, $\hat{\mathcal{T}}_{k,1}\neq \emptyset$ and $\hat{\mathcal{T}}_{k,2} \neq \emptyset$}  Since, $k\in \mathcal{N}_2$, we denote $\{p_1,p_2\}=\mathcal{P}_1(\mathsf{b}_k)$. Then, by partitioning the row indices of the polar encoding matrix, we obtain the following expression:
\begin{align}
	i_1(\mathsf{g}_i\oplus\mathsf{g}_j\oplus\mathsf{g}_k\oplus \mathsf{g}_{\hat{\mathcal{T}}})&=i_1(\mathsf{g}_i\oplus\mathsf{g}_j\oplus\mathsf{g}_k\oplus \mathsf{g}_{\hat{\mathcal{T}}}|\mathcal{M}^{c}_{p_1}\cap \mathcal{M}^{c}_{p_2})+i_1(\mathsf{g}_i\oplus\mathsf{g}_j\oplus\mathsf{g}_k\oplus \mathsf{g}_{\hat{\mathcal{T}}}|\mathcal{M}^{c}_{p_1}\cap \mathcal{M}_{p_2})\nonumber \\
	&\quad+ i_1(\mathsf{g}_i\oplus\mathsf{g}_j\oplus\mathsf{g}_k\oplus \mathsf{g}_{\hat{\mathcal{T}}}|\mathcal{M}_{p_1}\cap \mathcal{M}^{c}_{p_2})+ i_1(\mathsf{g}_i\oplus\mathsf{g}_j\oplus\mathsf{g}_k\oplus \mathsf{g}_{\hat{\mathcal{T}}}|\mathcal{M}_{p_1}\cap \mathcal{M}_{p_2})\nonumber \\
	&\overset{(a)}{=}i_1(\mathsf{g}^{\{p_1,p_2\}}_i\oplus \mathsf{g}^{\{p_1,p_2\}}_j\oplus \mathsf{g}^{\{p_1,p_2\}}_k \hspace{-0.15cm} \bigoplus_{\theta\in \hat{\mathcal{T}}_{k,1}}\hspace{-0.15cm}\mathsf{g}^{\{p_1,p_2\}}_{\theta}\bigoplus_{\theta\in \hat{\mathcal{T}}_{k,2}}\hspace{-0.15cm}\mathsf{g}^{\{p_1,p_2\}}_{\theta}\bigoplus_{\theta\in \tilde{\mathcal{T}}_0}\mathsf{g}^{\{p_1,p_2\}}_{\theta})\nonumber \\
	&+ i_1( \mathsf{g}^{\{p_1,p_2\}}_k \hspace{-0.15cm}\bigoplus_{\theta\in \hat{\mathcal{T}}_{k,1}}\hspace{-0.15cm}\mathsf{g}^{\{p_1,p_2\}}_{\theta}\bigoplus_{\theta\in \hat{\mathcal{T}}_{k,2}}\hspace{-0.15cm}\mathsf{g}^{\{p_1,p_2\}}_{\theta}\mathbb{I}\{b_{\theta,p_2}=1\})\nonumber\\
	&+ i_1(\mathsf{g}^{\{p_1,p_2\}}_k \hspace{-0.15cm}\bigoplus_{\theta\in \hat{\mathcal{T}}_{k,1}}\hspace{-0.15cm}\mathsf{g}^{\{p_1,p_2\}}_{\theta}\bigoplus_{\theta\in \hat{\mathcal{T}}_{k,2}}\hspace{-0.15cm}\mathsf{g}^{\{p_1,p_2\}}_{\theta}\mathbb{I}\{b_{\theta,p_1}=1\})+ i_1(\mathsf{g}^{\{p_1,p_2\}}_k \hspace{-0.15cm}\bigoplus_{\theta\in \hat{\mathcal{T}}_{k,1}}\hspace{-0.15cm}\mathsf{g}^{\{p_1,p_2\}}_{\theta}\})\nonumber\\
	&\overset{(b)}{\geq} i_1(\mathsf{g}^{\{p_1,p_2\}}_i \oplus \mathsf{g}^{\{p_1,p_2\}}_j \bigoplus_{\theta\in \hat{\mathcal{T}}_{k,2}}\hspace{-0.15cm}\mathsf{g}^{\{p_1,p_2\}}_{\theta}\mathbb{I}\{b_{\theta,p_2}=0\}\bigoplus_{\theta\in \tilde{\mathcal{T}}_0}\mathsf{g}^{\{p_1,p_2\}}_{\theta})\nonumber \\ 
	&+ i_1(\mathsf{g}^{\{p_1,p_2\}}_k \hspace{-0.15cm}\bigoplus_{\theta\in \hat{\mathcal{T}}_{k,1}}\hspace{-0.15cm}\mathsf{g}^{\{p_1,p_2\}}_{\theta}\bigoplus_{\theta\in \hat{\mathcal{T}}_{k,2}}\hspace{-0.15cm}\mathsf{g}^{\{p_1,p_2\}}_{\theta}\mathbb{I}\{b_{\theta,p_1}=1\})+ i_1(\mathsf{g}^{\{p_1,p_2\}}_k \hspace{-0.15cm}\bigoplus_{\theta\in \hat{\mathcal{T}}_{k,1}}\hspace{-0.15cm}\mathsf{g}^{\{p_1,p_2\}}_{\theta}\})
	\nonumber\\
	&\overset{(c)}{\geq} i_1(\mathsf{g}^{\{p_1,p_2\}}_i \oplus \mathsf{g}^{\{p_1,p_2\}}_j \bigoplus_{\theta\in \hat{\mathcal{T}}_{k,2}}\hspace{-0.15cm}\mathsf{g}^{\{p_1,p_2\}}_{\theta}\mathbb{I}\{b_{\theta,p_2}=0\} \bigoplus_{\theta\in \tilde{\mathcal{T}}_0}\mathsf{g}^{\{p_1,p_2\}}_{\theta} )\nonumber \\ 
	& +i_1(\bigoplus_{\theta\in \hat{\mathcal{T}}_{k,2}}\hspace{-0.15cm}\mathsf{g}^{\{p_1,p_2\}}_{\theta}\mathbb{I}\{b_{\theta,p_1}=1\})\overset{(d)}{\geq} 2^{\ell+1}
\end{align}
where (a) is due to \eqref{eq:subset_proj0} and \eqref{eq:subset_proj1}, (b) and (c) are due to the fact that $i_1(\mathsf{u}\oplus \mathsf{v})+i_1(\mathsf{v})\geq i_1(\mathsf{u})$, i.e., by \eqref{eq:u_v}, and (d) is by \eqref{eq:low_bound_d}, i.e., the Hamming weight of any row at each expression is at least $2^{\ell}$. This means that if neither $\hat{\mathcal{T}}_{k,1}$ nor $\hat{\mathcal{T}}_{k,2}$ is empty, the lower bound is satisfied whatever $\tilde{\mathcal{T}}_0$.

\paragraph*{Case~2, $\hat{\mathcal{T}}_{k,1}=\hat{\mathcal{T}}_{k,2}=\emptyset$} By partitioning the row indices with respect to $\{p_1,p_2\}=\mathcal{P}_1(\mathsf{b}_k)$, we obtain the following expression:
\begin{align}
	i_1(\mathsf{g}_i\oplus\mathsf{g}_j\oplus\mathsf{g}_k\oplus \mathsf{g}_{\mathcal{T}})&=i_1(\mathsf{g}_i\oplus\mathsf{g}_j\oplus\mathsf{g}_k\oplus \mathsf{g}_{\mathcal{T}}|\mathcal{M}^{c}_{p_1}\cap \mathcal{M}^{c}_{p_2})+i_1(\mathsf{g}_i\oplus\mathsf{g}_j\oplus\mathsf{g}_k\oplus \mathsf{g}_{\mathcal{T}}|\mathcal{M}^{c}_{p_1}\cap \mathcal{M}_{p_2})\nonumber \\
	&\quad+ i_1(\mathsf{g}_i\oplus\mathsf{g}_j\oplus\mathsf{g}_k\oplus \mathsf{g}_{\mathcal{T}}|\mathcal{M}_{p_1}\cap \mathcal{M}^{c}_{p_2})+ i_1(\mathsf{g}_i\oplus\mathsf{g}_j\oplus\mathsf{g}_k\oplus \mathsf{g}_{\mathcal{T}}|\mathcal{M}_{p_1}\cap \mathcal{M}_{p_2})\nonumber \\
	&=\underbrace{i_1(\mathsf{g}^{\{p_1,p_2\}}_i\oplus \mathsf{g}^{\{p_1,p_2\}}_j\oplus \mathsf{g}^{\{p_1,p_2\}}_k \hspace{-0.15cm} \bigoplus_{\theta\in \tilde{\mathcal{T}}_0}\mathsf{g}^{\{p_1,p_2\}}_{\theta})+i_1( \mathsf{g}^{\{p_1,p_2\}}_k )}_{ \geq i_1(\mathsf{g}^{\{p_1,p_2\}}_i\oplus \mathsf{g}^{\{p_1,p_2\}}_j \bigoplus_{\theta\in \tilde{\mathcal{T}}_0}\mathsf{g}^{\{p_1,p_2\}}_{\theta}) \text{ by }\eqref{eq:u_v}} + 2\cdot i_1( \mathsf{g}^{\{p_1,p_2\}}_k )\nonumber\\
	&\geq i_1(\mathsf{g}^{\{p_1,p_2\}}_i\oplus \mathsf{g}^{\{p_1,p_2\}}_j \bigoplus_{\theta\in \tilde{\mathcal{T}}_0}\mathsf{g}^{\{p_1,p_2\}}_{\theta})+2\cdot i_1( \mathsf{g}^{\{p_1,p_2\}}_k )\overset{(a)}{=} 2^{\ell+1}
\end{align}
where (a) comes from the fact that the rows $\mathsf{g}^{\{p_1,p_2\}}_i$, $\mathsf{g}^{\{p_1,p_2\}}_j$ and $\mathsf{g}^{\{p_1,p_2\}}_{\tilde{\mathcal{T}}_0}$ comply with the conditions of Theorem~\ref{thm:main_min_dis_inc}. Hence,  $i_1(\mathsf{g}^{\{p_1,p_2\}}_{\{i,j,\tilde{\mathcal{T}}_0\}})\geq 2^{\ell+1}-2$, and $2\cdot i_1( \mathsf{g}^{\{p_1,p_2\}}_k )=2$.


The same result can be shown for the other cases of $\hat{\mathcal{T}}_{k,1}, \hat{\mathcal{T}}_{k,2}$ by following similar steps.

\subsection{Proof of Theorem~\ref{thm:Last_ensemble}}\label{app:thm4}

The proof of this theorem is based on Theorem~\ref{thm:main_row_merg0} and Theorem~\ref{thm:main_min_dis_inc_with_intrsctn_two_rows}. The latter is a generalization of Theorem~\ref{thm:main_min_dis_inc} for the case where there are some common 1-bit indices in the intersection of the binary representations of the pair.


For $m= t_0+t_1+1$, the code $\bar{\mathbf{C}}$ is given as
\begin{align}
	\bar{\mathbf{C}}=\{\mathbf{C}\}\hspace{-0.3cm}\bigcup_{\mathcal{D}\subseteq [0,t_0+t_1]}\hspace{-0.5cm}\{\mathsf{c}:\mathsf{c}=\bigoplus_{\theta \in \mathcal{D}}\mathsf{g}_{\{\Pi_{i}^{\theta},\Pi_{j}^{\theta},\Pi_{k}^{\theta} \}}\oplus \mathsf{g}_{\mathcal{T}}, \, \mathcal{T}\subseteq \mathcal{A}\}
\end{align}
and hence, it is sufficient to prove the following statement
\begin{align}
	i_1(\bigoplus_{\theta \in \mathcal{D}}\mathsf{g}_{\{\Pi_{i}^{\theta},\Pi_{j}^{\theta},\Pi_{k}^{\theta} \}}\oplus \mathsf{g}_{\mathcal{T}})\geq 2^{\ell+1}, \; \; \; 
\end{align}
for any $\mathcal{D}\subseteq [0,t_0+t_1]\text{ and }\mathcal{T}\subseteq \mathcal{A}$.

The proof is done for the whole set $\{\mathsf{g}_{\Pi^{d}_{i}}\oplus\mathsf{g}_{\Pi^{d}_{j}}\oplus\mathsf{g}_{\Pi^{d}_{k}}\}$, $d\in [0,t_0+t_1] $ and the case $i_0(\mathsf{b}_i\bar{\cup} \mathsf{b}_j\bar{\cup} \mathsf{b}_k)>i_1(\mathsf{b}_i\bar{\cap} \mathsf{b}_j\bar{\cap} \mathsf{b}_k)$. 
The same result can be found by following similar steps for any combination of $\{\mathsf{g}_{\Pi^{d}_{i}}\oplus\mathsf{g}_{\Pi^{d}_{j}}\oplus\mathsf{g}_{\Pi^{d}_{k}}\}$, $d\in [0,t_0+t_1] $ and other cases such as $i_0(\mathsf{b}_i\bar{\cup} \mathsf{b}_j\bar{\cup} \mathsf{b}_k)=i_1(\mathsf{b}_i\bar{\cap} \mathsf{b}_j\bar{\cap} \mathsf{b}_k)$ or $i_0(\mathsf{b}_i\bar{\cup} \mathsf{b}_j\bar{\cup} \mathsf{b}_k)=0$. 

For the considered case, the main steps are summarized in \eqref{eq_topofpage} on top of the next page, where (a) is due to the fact that $\mathcal{M}^{c}_{n-2}\cap \mathcal{M}_{n-2}= \emptyset$, (b) is due to \eqref{eq:proj_v1}, (c) is due to the fact that the minimum Hamming weight of any row involved in the right-hand side of the equation is $2^{\ell-1}$, (d) is due to Theorem~\ref{thm:main_row_merg0}, (e) is due to \eqref{eq:set_greater1}, \eqref{eq:set_greater2} and \eqref{eq:set_greater3}, (f) is since $\Pi^{t_1}_{i},\Pi^{t_1}_{j}$ and $\{t: b_{t,n-1}\!=\!b_{t,t_0}\!=\!\cdots\!=\!b_{t,t_1}\!=\!0, t\in \mathcal{T}\}$ comply with the conditions of Theorem~\ref{thm:main_min_dis_inc_with_intrsctn_two_rows}, (g) is due to Theorem~\ref{th:Hamm_of_sum0000} and (h) is due to the assumption $i_1(\mathsf{b}_i\bar{\cap}\mathsf{b}_j)\leq \ell-2$.

The proof can be conducted the same way for all $m < t_0 + t_1 + 1$ and the proof is complete.
\begin{figure*}
	\begin{align}
		&i_1\left(\bigoplus_{d=0}^{t_0+t_1}\mathsf{g}_{\{\Pi^{d}_{i},\Pi^{d}_{j},\Pi^{d}_{k}\}}\oplus \mathsf{g}_{\mathcal{T}}\right)\overset{(a)}{=}i_1\left(\bigoplus_{d=0}^{t_0+t_1}\mathsf{g}_{\{\Pi^{d}_{i},\Pi^{d}_{j},\Pi^{d}_{k}\}}\oplus \mathsf{g}_{\mathcal{T}}|\mathcal{M}^{c}_{n-2}\right)+i_1\left(\bigoplus_{d=0}^{t_0+t_1}\mathsf{g}_{\{\Pi^{d}_{i},\Pi^{d}_{j},\Pi^{d}_{k}\}}\oplus \mathsf{g}_{\mathcal{T}}|\mathcal{M}_{n-2}\right)\nonumber \\
		&\qquad \overset{(b)}{=}i_1\left(\bigoplus_{d=0}^{t_0+t_1}\mathsf{g}^{n-2}_{\{\Pi^{d}_{i},\Pi^{d}_{j},\Pi^{d}_{k}\}} \oplus \mathsf{g}^{n-2}_{\mathcal{T}}\right)+i_1\left(\bigoplus_{d=0}^{t_0+t_1}\bigoplus_{t\in\{\Pi^{d}_{i},\Pi^{d}_{j},\Pi^{d}_{k}\}}\mathsf{g}^{n-2}_{t}\mathbb{I}\{b_{t,n-2}=1\}\bigoplus_{t\in \mathcal{T}}\mathsf{g}^{n-2}_{t}\mathbb{I}\{b_{t,n-2}=1\}\right) \nonumber \\ &\qquad\overset{(c)}{\geq} i_1\left(\bigoplus_{d=0}^{t_0+t_1}\mathsf{g}^{n-2}_{\{\Pi^{d}_{i},\Pi^{d}_{j},\Pi^{d}_{k}\}} \oplus \mathsf{g}^{n-2}_{\mathcal{T}}\right)+ 2^{\ell-1}\nonumber \\ 
		&\qquad\overset{(d)}{\geq}i_1\left(\bigoplus_{d=0}^{t_0+t_1}\hspace{-0.1cm}\bigoplus_{t\in\{\Pi^{d}_{i},\Pi^{d}_{j},\Pi^{d}_{k}\}} \hspace{-0.6cm}\mathsf{g}^{n-2}_{t}\mathbb{I}\{b_{t,n-1}\!=\!b_{t,t_0}\!=\!\cdots\!=\!b_{t,t_1}\!=\!0\}\bigoplus_{t\in \mathcal{T}}\mathsf{g}^{n-2}_{t}\mathbb{I}\{b_{t,n-1}\!=\!b_{t,t_0}\!=\!\cdots\!=\!b_{t,t_1}\!=\!0\} \right) \nonumber \\ & \hspace{15cm}  + 2^{\ell-1} \nonumber \\ 
		&\qquad\overset{(e)}{=} i_1\left( \mathsf{g}^{n-2}_{\Pi^{t_1}_{i}} \oplus \mathsf{g}^{n-2}_{\Pi^{t_1}_{j}}\bigoplus_{t\in \mathcal{T}}\mathsf{g}^{n-2}_{t}\mathbb{I}\{b_{t,n-1}\!=\!b_{t,t_0}\!=\!\cdots\!=\!b_{t,t_1}\!=\!0\}\right) + 2^{\ell-1} \nonumber \\
		&\qquad\overset{(f)}{\geq} i_1\left(\mathsf{g}^{n-2}_{\Pi^{t_1}_{i}} \oplus \mathsf{g}^{n-2}_{\Pi^{t_1}_{j}}\right) + 2^{\ell-1} \overset{(g)}{=} 2^{\ell+1}-2^{i_1(\mathsf{b}_i\bar{\cap} \mathsf{b}_j)+1}+2^{\ell-1}\overset{(h)}{\geq} 2^{\ell+1}\label{eq_topofpage}
	\end{align}
	\hrulefill
\end{figure*}



\end{document}